\documentclass{llncs}
\usepackage{amsmath}
\usepackage{amsfonts}
\usepackage{amssymb}
\usepackage{mathtools}
\usepackage{hyperref}

\usepackage{algorithm}
\usepackage{algpseudocode}
\usepackage[utf8]{inputenc}
\usepackage{tikz}
\usepackage{pgfplots}
\usepackage{bbm}
\usepackage{multirow}
\usepackage{enumerate}
\usepackage{enumitem}
\usepackage{booktabs}
\usepackage{xcolor}
\usepackage{xifthen}
\usepackage{tikz}
\usetikzlibrary{trees}
\usepackage{dsfont}

\allowdisplaybreaks

\usepackage{calc}
\usepackage{setspace}
\newcommand{\undertext}[3][0]{\underset{\text{\parbox{\widthof{$#3$} + #1cm}{\small \centering \setstretch{0.7}{#2}}}}{\underbrace{#3}}}

\pgfplotscreateplotcyclelist{my black white}{%
solid, every mark/.append style={solid, fill=gray}, mark=*\\%
dotted, every mark/.append style={solid, fill=gray}, mark=square*\\%
densely dotted, every mark/.append style={solid, fill=gray}, mark=otimes*\\%
loosely dotted, every mark/.append style={solid, fill=gray}, mark=triangle*\\%
dashed, every mark/.append style={solid, fill=gray},mark=diamond*\\%
}

\renewcommand{\vec}[1]{\ensuremath{\mathbf{#1}}}

\newcommand{\entr}[1]{h(#1)}
\newcommand{\bin}[2]{\mathrm{bin}\left(#1,#2\right)}
\newcommand{\trin}[3]{\mathrm{trin}\left(#1,#2,#3\right)}
\newcommand{\quadrin}[4]{\mathrm{quadrin}\left(#1,#2,#3,#4\right)}

\newcommand{\bigO}[1]{ \mathcal{O} \left(#1 \right)}
\newcommand{\bigOt}[1]{ \widetilde{\mathcal{O}} \left(#1 \right)}
\newcommand{\zo}{\{0,1\}}
\newcommand{\zom}{\{-1,0,1\}}
\newcommand{\zomd}{\{-1,0,1,2\}}

\newcommand{\ket}[1]{\left| #1 \right\rangle}

\newcommand{\poly}{\operatorname{poly}}
\newcommand{\negl}{\operatorname{negl}}
\newcommand{\card}[1]{\left|#1\right|}
\newcommand{\dotprod}{\boldsymbol{\cdot}}

\newcommand{\Exp}[2]{\mathbb{E}_{#1}\bracket{#2}}
\newcommand{\Ber}{\operatorname{Ber}}

\newcommand{\COMMENT}[1]{}

\newcommand{\pfilterhgj}[2]{\mathsf{pf}_1 \left(#1, #2 \right)}
\newcommand{\pfilterbcj}[3]{\mathsf{pf}_2 \left(#1, #2, #3 \right)}
\newcommand{\pfilternous}[5]{\mathsf{pf}_3 \left(#1, #2, #3, #4, #5 \right)}
\newcommand{\pf}{\mathsf{pf}}

\usepackage{xifthen}

\newcommand{\prob}[2][]{\Pr
\ifthenelse{\isempty{#1}}{}{_{#1}}
\ifthenelse{\isempty{#2}}{}{\left[ #2 \right]}
}

\DeclarePairedDelimiter{\set}{\{}{\}}
\DeclarePairedDelimiter{\bracket}{[}{]}

\newif\ifeprint
\eprintfalse

\newcommand{\AS}[1]{\textcolor{red}{[{\bf Andr\'{e} S:} #1]}}
\newcommand{\XB}[1]{\textcolor{blue}{[{\bf Xavier B:} #1]}}
\newcommand{\YS}[1]{\textcolor{green}{[{\bf Yixin S:} #1]}}
\newcommand{\RB}[1]{\textcolor{orange}{[{\bf Rémi B:} #1]}}
 \renewcommand{\AS}[1]{}
 \renewcommand{\XB}[1]{}
 \renewcommand{\YS}[1]{}
 \renewcommand{\RB}[1]{}

\newtheorem{heuristic}{Heuristic}
\newtheorem{fact}{Fact}

\algblockdefx[Repeat]{Repeat}{EndRepeat}%
[1]{\textbf{Repeat} #1 \textbf{times}}%
{\textbf{EndRepeat}}

\title{Improved Classical and Quantum Algorithms for Subset-Sum}
\author{Xavier Bonnetain\inst{1} \and Rémi Bricout\inst{2,3} \and André Schrottenloher\inst{3} \and Yixin Shen\inst{4}}
\institute{Institute for Quantum Computing, Department of Combinatorics and Optimization, University of Waterloo, Waterloo, ON, Canada \and Sorbonne Université, Collège Doctoral, F-75005 Paris, France\and Inria, France \and  Université de Paris, IRIF, CNRS, F-75006 Paris, France}

\begin{document}
\maketitle
\renewcommand{\labelitemi}{$\bullet$}

\begin{abstract}
We present new classical and quantum algorithms for solving random subset-sum instances. First, we improve over the Becker-Coron-Joux algorithm (EUROCRYPT 2011) from $\bigOt{2^{0.291 n}}$ down to $\bigOt{2^{0.283 n}}$, using more general representations with values in $\zomd$.

Next, we improve the state of the art of quantum algorithms for this problem in several directions. By combining the Howgrave-Graham-Joux algorithm (EUROCRYPT 2010) and quantum search, we devise an algorithm with asymptotic running time $\bigOt{2^{0.236 n}}$, lower than the cost of the quantum walk based on the same classical algorithm proposed by Bernstein, Jeffery, Lange and Meurer (PQCRYPTO 2013). This algorithm has the advantage of using \emph{classical} memory with quantum random access, while the previously known algorithms used the quantum walk framework, and required \emph{quantum} memory with quantum random access.

We also propose new quantum walks for subset-sum, performing better than the previous best time complexity of $\bigOt{2^{0.226 n}}$ given by Helm and May (TQC 2018). We combine our new techniques to reach a time $\bigOt{2^{0.216 n}}$. This time is dependent on a heuristic on quantum walk updates, formalized by Helm and May, that is also required by the previous algorithms. We show how to partially overcome this heuristic, and we obtain an algorithm with quantum time $\bigOt{2^{0.218 n}}$ requiring only the standard classical subset-sum heuristics.


\end{abstract}


\section{Introduction}
\label{section:intro}

We study the \emph{subset-sum problem}, also known as \emph{knapsack problem}: given $n$ integers $\vec{a} = (a_1, \ldots a_n)$, and a target integer $S$, find an $n$-bit vector $\vec{e} = (e_1, \ldots e_n) \in \zo^n$ such that $\vec{e} \cdot \vec{a} = \sum_i e_i a_i = S$. The \emph{density} of the knapsack instance is defined as $d = n / (\log_2 \max_i a_i)$, and for a random instance $\vec{a}$, it is related to the number of solutions that one can expect.

The decision version of the knapsack problem is NP-complete~\cite{DBLP:books/fm/GareyJ79}. Although certain densities admit efficient algorithms, related to lattice reduction~\cite{DBLP:conf/focs/LagariasO83,DBLP:conf/approx/Lyubashevsky05}, the best algorithms known for the knapsack problem when the density is close to 1 are  exponential-time, which is why we name these instances ``hard'' knapsacks. This problem underlies some cryptographic schemes aiming at post-quantum security (see \emph{e.g.}~\cite{DBLP:conf/tcc/LyubashevskyPS10}), 
and is used as a building block in some quantum hidden shift algorithms~\cite{Arxiv:Bonnetain19}, which have some applications in quantum cryptanalysis of isogeny-based~\cite{BonSch20} and symmetric cryptographic schemes~\cite{DBLP:conf/asiacrypt/BonnetainN18}. 

In this paper, we focus on the case where $d = 1$, where expectedly a single solution exists. Instead of naively looking for the solution $\vec{e}$ via exhaustive search, in time $2^n$, Horowitz and Sahni~\cite{DBLP:journals/jacm/HorowitzS74} proposed to use a meet-in-the-middle approach in $2^{n/2}$ time and memory. The idea is to find a collision between two lists of $2^{n/2}$ subknapsacks, \emph{i.e.} to merge these two lists for a single solution. Schroeppel and Shamir~\cite{DBLP:journals/siamcomp/SchroeppelS81} later improved this to a 4-list merge, in which the memory complexity can be reduced down to $2^{n/4}$.

\paragraph{The Representation Technique.}
At EUROCRYPT 2010, Howgrave-Graham and Joux~\cite{DBLP:conf/eurocrypt/Howgrave-GrahamJ10} (HGJ) proposed a heuristic algorithm solving \emph{random} subset-sum instances in time $\bigOt{2^{0.337 n}}$, thereby breaking the $2^{n/2}$ bound. Their key idea was to represent the knapsack solution ambiguously as a sum of vectors in $\zo^n$. This \emph{representation technique} increases the search space size, allowing to merge more lists, with new arbitrary  constraints, thereby allowing for a more time-efficient algorithm. The time complexity exponent is obtained by numerical optimization of the list sizes and constraints, assuming that the individual elements obtained in the merging steps are well-distributed. This is the standard heuristic of classical and quantum subset-sum algorithms. Later, Becker, Coron and Joux~\cite{DBLP:conf/eurocrypt/BeckerCJ11} (BCJ) improved the asymptotic runtime down to $\bigOt{2^{0.291 n}}$ by allowing even more representations, with vectors in $\zom^n$.

The BCJ representation technique is not only a tool for subset-sums, as it has been used to speed up generic decoding algorithms, classically~\cite{DBLP:conf/asiacrypt/MayMT11,DBLP:conf/eurocrypt/BeckerJMM12,DBLP:conf/eurocrypt/0001O15} and quantumly~\cite{DBLP:conf/pqcrypto/KachigarT17}. Therefore, the subset-sum problem serves as the simplest application of representations, and improving our understanding of the classical and quantum algorithms may have consequences on these other generic problems.

\paragraph{Quantum Algorithms for the Subset-Sum Problem.}
Cryptosystems based on hard subset-sums are natural candidates for post-quantum cryptography, but to understand precisely their security, we have to study the best generic algorithms for solving subset-sums. The first quantum time speedup for this problem was obtained in~\cite{DBLP:conf/pqcrypto/BernsteinJLM13}, with a quantum time $\bigOt{2^{0.241 n}}$. The algorithm was based on the HGJ algorithm. Later on,~\cite{DBLP:conf/tqc/HelmM18} devised an algorithm based on BCJ, running in time $\bigOt{2^{0.226 n}}$. Both algorithms use the corresponding classical merging structure, wrapped in a quantum walk on a Johnson graph, in the MNRS quantum walk framework ~\cite{mnrs11}. However, they suffer from two limitations. 

First, both use the model of \emph{quantum memory with quantum random-access} (QRAQM), which is stronger than the standard quantum circuit model, as it allows unit-time lookups in superposition of all the qubits in the circuit. The QRAQM model is used in most\XB{Most?} quantum walk algorithms to date, but its practical realizations are still unclear. With a more restrictive model, i.e. \emph{classical memory with quantum random-access} (QRACM), no quantum time speedup over BCJ was previously known. This is not the case for some other hard problems in post-quantum cryptography, \emph{e.g.} heuristic lattice sieving for the Shortest Vector Problem, where the best quantum algorithms to date require only QRACM~\cite{laarhoven2015search}.

Second, both use a conjecture (implicit in~\cite{DBLP:conf/pqcrypto/BernsteinJLM13}, made explicit in~\cite{DBLP:conf/tqc/HelmM18}) about quantum walk updates. In short, the quantum walk maintains a data structure, that contains a merging tree similar to HGJ (resp. BCJ), with lists of smaller size. A quantum walk step is made of updates that changes an element in the lowest-level lists, and requires to modify the upper levels accordingly, \emph{i.e.} to track the partial collisions that must be removed or added. In order to be efficient, the update needs to run in polynomial time. Moreover, the resulting data structure shall be a function of the lowest-level list, and not depend on the path taken in the walk. The conjecture states that it should be possible to guarantee sound updates without impacting the time complexity exponent. However, it does not seem an easy task and the current literature on subset-sums lacks further justification or workarounds.




\paragraph{Contributions.}
In this paper, we improve classical and quantum subset-sum algorithms based on representations. We write these algorithms as sequences of ``merge-and-filter'' operations, where lists of subknapsacks are first merged with respect to an arbitrary constraint, then \emph{filtered} to remove the subknapsacks that cannot be part of a solution.

First, we propose a more time-efficient classical subset-sum algorithm based on representations. We have two classical improvements: we revisit the previous algorithms and show that some of the constraints they enforced were not needed, and we use more general distributions by allowing ``2''s in the representations. Overall, we obtain a better time complexity exponent of $0.283$.


Most of our contributions concern quantum algorithms. As a generic tool, we introduce \emph{quantum filtering}, which speeds up the filtering of representations with a quantum search. We use this improvement in all our new quantum algorithms.

We give an improved quantum walk based on quantum filtering and our extended $\{-1,0,1,2\}$ representations. Our best runtime exponent is $0.216$, under the quantum walk update heuristic of~\cite{DBLP:conf/tqc/HelmM18}. Next, we show how to overcome this heuristic, by designing a new data structure for the vertices in the quantum walk, and a new update procedure with guaranteed time. We remove this heuristic from the previous algorithms~\cite{DBLP:conf/pqcrypto/BernsteinJLM13,DBLP:conf/tqc/HelmM18} with no additional cost. However, we find that removing it from our quantum walk increases its cost to $0.218$.

In a different direction, we devise a new quantum subset-sum algorithm based on HGJ, with time $\bigOt{2^{0.236 n}}$. It is the first quantum time speedup on subset-sums that is \emph{not} based on a quantum walk. The algorithm performs instead a depth-first traversal of the HGJ tree, using quantum search as its only building block. Hence, by construction, it does not require the additional heuristic of~\cite{DBLP:conf/tqc/HelmM18} \emph{and} it only uses \emph{classical memory with quantum random-access}, giving also the first quantum time speedup for subset-sum in this memory model.

A summary of our contributions is given in Table~\ref{table:results}\footnote{After this work, Alexander May has informed us that the thesis~\cite{bohme11} contains unpublished results using more symbols, with the best exponent of 0.2871 obtained with the symbol set $\{-2, -1,0,1,2\}$. }. 
All these complexity exponents are obtained by numerical optimization. Our code is available at \url{https://github.com/xbonnetain/optimization-subset-sum}.


\begin{table}
\centering
\caption{Previous and {\bf new} algorithms for subset-sum, classical and quantum, with time and memory exponents rounded upwards. We note that the removal of Heuristic~\ref{heuristic:helm-may} in~\cite{DBLP:conf/pqcrypto/BernsteinJLM13,DBLP:conf/tqc/HelmM18} comes from our new analysis in Section~\ref{subsection:time-without-heuristics}. QW: Quantum Walk. QS: Quantum Search. CF: Constraint filtering (not studied in this paper). QF: Quantum filtering.}
\label{table:results}

\begin{tabular}{ccccccc}
\toprule
 \begin{tabular}{c}  Time \\ exp. \end{tabular} &  \begin{tabular}{c}  Memory \\ exp. \end{tabular}  & \begin{tabular}{c}Represen-\\tations\end{tabular} & \begin{tabular}{c}
 Memory \\ model
\end{tabular} & Techniques &\begin{tabular}{c}
Requires \\
 Heur.~\ref{heuristic:helm-may}
\end{tabular}  & Reference \\
\midrule
\multicolumn{6}{c}{Classical} \\
\midrule
 0.3370 & 0.3113 & $\zo$ & RAM & &   & \cite{DBLP:conf/eurocrypt/Howgrave-GrahamJ10} \\
0.2909 & 0.2909 & $\zom$ & RAM &  & & \cite{DBLP:conf/eurocrypt/BeckerCJ11} \\
0.287 &  & $\zom$ & RAM & CF && \cite{ozerovphd} \\
\bf 0.2830 & \bf  0.2830 &  $\zomd$ & RAM & &  & Sec.~\ref{subsection:new-classical-bcj} \\
 \midrule
 \multicolumn{6}{c}{Quantum} \\
 \midrule
0.241 & 0.241 & $\zo$ & QRAQM & QW & \bf No & \cite{DBLP:conf/pqcrypto/BernsteinJLM13} + Sec.~\ref{subsection:time-without-heuristics}  \\
0.226 & 0.226 & $\zom$ & QRAQM & QW & \bf No & \cite{DBLP:conf/tqc/HelmM18} + Sec.~\ref{subsection:time-without-heuristics} \\
\bf 0.2356 & \bf 0.2356 & $\zo$ & \bf QRACM &QS + QF & No & Sec~\ref{subsection:hgj-quantum-second} \\
\bf 0.2156 & \bf 0.2110 & $\zomd$ & QRAQM & QW + QF &Yes & Sec.~\ref{subsection:new-quantum-bcj} \\
\bf 0.2182 & \bf 0.2182 & $\zomd$ & QRAQM & QW + QF &No & Sec.~\ref{subsection:time-without-heuristics} \\
\bottomrule
\end{tabular}
\end{table}



\paragraph{Outline.}
In Section~\ref{section:classical-algos}, we study classical algorithms. We review the representation technique, the HGJ algorithm and introduce our new $\{-1,0,1,2\}$ representations to improve over~\cite{DBLP:conf/eurocrypt/BeckerCJ11}. In Section~\ref{section:quantum-prelim}, we move to the quantum setting, introduce some preliminaries and the previous quantum algorithms for subset-sum. In Section~\ref{section:hgj-asymmetric}, we present and study our new quantum algorithm based on HGJ and quantum search. We give different optimizations and time-memory tradeoffs. In Section~\ref{section:new-qw}, we present our new quantum algorithm based on a quantum walk. Finally, in Section~\ref{section:fix-qw} we show how to overcome the quantum walk update conjecture, up to a potential increase in the update cost. We conclude, and give a summary of our new results in Section~\ref{section:conclusion}.

\section{List Merging and Classical Subset-sum Algorithms}
\label{section:classical-algos}

In this section, we remain in the classical realm. We introduce the standard subset-sum notations and heuristics and give a new presentation of the HGJ algorithm, putting an emphasis on the \emph{merge-and-filter} operation. We introduce our extended $\{-1,0,1,2\}$ representations and detail our improvements over BCJ.

\subsection{Notations and Conventions}


Hereafter and in the rest of the paper, all time and memory complexities, classical and quantum, are exponential in $n$. We use the soft-O notation $\widetilde{\mathcal{O}}$ which removes polynomial factors in $n$, and focus on the asymptotic exponent, relative to $n$. We use $\negl(n)$ for any function that vanishes inverse-exponentially in $n$. We often replace asymptotic exponential time and memory complexities (\emph{e.g.} $\bigOt{2^{\alpha n}}$) by their exponents (\emph{e.g.} $\alpha$). We use capital letters (\emph{e.g.} $L$) and corresponding letters (\emph{e.g.} $\ell$) to denote the same value, in $\log_2$ and relatively to $n$: $\ell = \log_2(L) / n$. 

\begin{definition}[Entropies and multinomial functions]
We define the following functions:
 \begin{enumerate}[topsep=0pt,leftmargin=3.6cm]
  \item[\bf Hamming Entropy:]$\entr{x} = - x \log_2 x - (1-x) \log_2 (1-x)$
  \item[\bf Binomial:]$\bin{\omega}{\alpha}  = \entr{\alpha / \omega} \omega$
  \item[\bf 2-way Entropy:]$g(x,y) = -x \log_2 x - y \log_2 y - (1-x-y) \log_2 (1-x-y)$
  \item[\bf Trinomial:]$\trin{\omega}{\alpha}{\beta} = g(\alpha / \omega, \beta / \omega) \omega$
  \item[\bf 3-way Entropy:]$f(x, y, z) =  - x \log_2 x - y \log_2 y - z \log_2 z -$\\
  \phantom{a} \hfill $(1-x-y-z) \log_2 (1-x-y-z)$
  \item[\bf Quadrinomial:] $\quadrin{\omega}{\alpha}{\beta}{\gamma} = f(\alpha/\omega, \beta/\omega, \gamma/\omega)\omega$
 \end{enumerate}
\end{definition}

\begin{property}[Standard approximations]%
 We have the following approximations, asymptotically in $n$:
\begin{center}$\bin{\omega}{\alpha} \simeq \frac{1}{n} \log_2{\omega n \choose \alpha n} ~~;~~ \trin{\omega}{\alpha}{\beta} \simeq \frac{1}{n} \log_2{\omega n \choose \alpha n, \beta n} $ $ \quadrin{\omega}{\alpha}{\beta}{\gamma} \simeq \frac{1}{n} \log_2 { \omega n \choose \alpha n, \beta n, \gamma n }$
\end{center}

\end{property}

%

\begin{definition}[Distributions of knapsacks]
A \emph{knapsack} or \emph{subknapsack} is a vector $\vec{e} \in \zomd^n$. The set of $\vec{e}$ with $\alpha n$ ``-1'', $(\alpha + \beta - 2\gamma) n$ ``1'', $\gamma n$ ``2'' and $(1 - 2\alpha - \beta + \gamma)n$ ``0'' is denoted $D^n[\alpha, \beta, \gamma]$. 
If $\gamma = 0$, we may omit the third parameter. This coincides with the notation $D^n[\alpha, \beta]$ from~\cite{DBLP:conf/tqc/HelmM18}.
\end{definition}

Note that we always add vectors \emph{over the integers}, and thus, the sum of two vectors of $D^n[*,*,*]$ may contain unwanted symbols $-2, 3$ or $4$.

\begin{property}[Size of knapsack sets]
\label{property:knapsack_size}
We have:
\begin{center}
$\frac{1}{n} \log_2 | D^n[0, \beta,0] | \simeq h(\beta) ~~; ~~\frac{1}{n} \log_2 | D^n[\alpha, \beta, 0] | \simeq  g(\alpha, \alpha + \beta)$
\vspace{1em}

$ \frac{1}{n} \log_2 | D^n[\alpha, \beta, \gamma] | \simeq f(\alpha, \alpha + \beta - 2\gamma, \gamma) \enspace. $
\end{center}

\end{property}

\paragraph{Subset-sum.}The problem we will solve is defined as follows:
\begin{definition}[Random subset-sum instance of weight {$n/2$}]
Let $\vec{a}$ be chosen uniformly at random from $\left(\mathbb{Z}_{N}\right)^n$, where $N \simeq 2^n$. Let $\vec{e}$ be chosen uniformly at random from $D^n[0, 1/2, 0]$. Let $t = \vec{a} \cdot \vec{e} \pmod{N}$. Then $(\vec{a}, t)$ is a \emph{random subset-sum instance}. A \emph{solution} is a vector $\vec{e'}$ such that $\vec{a} \cdot \vec{e'} = t \pmod{N}$.
\end{definition}


\paragraph{Sampling.}
Throughout this paper, we assume that we can classically sample uniformly at random from $D^n[\alpha, \beta, \gamma]$ in time $\poly(n)$. (Since $\alpha n$, $\beta n$ and $\gamma n$ will in general not be integer, we suppose to have them rounded to the nearest integer.)
This comes from an efficient bijection between representations and integers (see \ifeprint Appendix~A\else Appendix A in the full version of the paper~\cite{DBLP:journals/iacr/BonnetainBSS20}\fi).
In addition, we can efficiently produce the uniform superposition of vectors of $D^n[\alpha, \beta, \gamma]$, using $\poly(n)$ quantum gates, and we can perform a quantum search among representations.




\subsection{Merging and Filtering}
\label{subsection:merging-filtering}

In all subset-sum algorithms studied in this paper, we repeatedly sample vectors with certain distributions $D^n[*, *, *]$, then combine them. Let $D_1 = D^n[\alpha_1, \beta_1, \gamma_1]$, $D_2 = D^n[\alpha_2, \beta_2, \gamma_2]$ be two input distributions and $D = D^n[\alpha, \beta, \gamma]$ be a target. Given two lists $L_1 \in D_1^{|L_1|}$ and $L_2 \in D_2^{|L_2|}$, we define:
\begin{itemize}
\item the \emph{merged list} $L = L_1 \bowtie_c L_2$ containing all vectors $\vec{e} = \vec{e_1} + \vec{e_2}$ such that: $\vec{e_1} \in L_1, \vec{e_2} \in L_2$, $(\vec{e}_1 + \vec{e}_2) \cdot \vec{a} = s \mod{M}$, $s \leq M$ is an arbitrary integer and $M \approx 2^{cn}$ (we write $L_1 \bowtie_c L_2$ because $s$ is an arbitrary value, whose choice is without incidence on the algorithm)
\item the \emph{filtered list} $L^f = (L \cap D) \subseteq L$, containing the vectors with the target distribution of $1, -1, 2$ (the target $D$ will always be clear from context).
\end{itemize}
In general, $L$ is exponentially bigger than $L^f$ and does not need to be written down, as vectors can be filtered on the fly. The algorithms then repeat the merge-and-filter operation on multiple levels, moving towards the distribution $D^n[0, 1/2]$ while increasing the bit-length of the modular constraint, until we satisfy $\vec{e} \cdot \vec{a} = t \mod{2^n}$ and obtain a solution. Note that this merging-and-filtering view that we adopt, where the merged list is repeatedly sampled before an element passes the filter, has some similarities with the ideas developed in the withdrawn article~\cite{DBLP:journals/corr/abs-1907-04295}.

The standard subset-sum heuristic assumes that vectors in $L^f$ are drawn independently, uniformly at random from $D$. It simplifies the complexity analysis of both classical and quantum algorithms studied in this paper. Note that this heuristic, which is backed by experiments, actually leads to provable probabilistic algorithms in the classical setting (see~\cite[Theorem 2]{DBLP:conf/eurocrypt/BeckerCJ11}). We adopt the version of~\cite{DBLP:conf/tqc/HelmM18}.

\begin{heuristic}
\label{heur1}
If input vectors are uniformly distributed in $D_1 \times D_2$, then the filtered pairs are uniformly distributed in $D$ (more precisely, among the subset of vectors in $D$ satisfying the modular condition).
\end{heuristic}

\paragraph{Filtering Representations.} We note $\ell = (1/n)\log_2|L|$, and so on for $\ell_1, \ell_2, \ell^f$.
By Heuristic~\ref{heur1}, the average sizes of $L_1$, $L_2$, $L$ and $L^f$ are related by:
\begin{itemize}
\item $\ell = \ell_1 + \ell_2 - c$
\item $\ell^f = \ell + \pf$, where $\pf$ is negative and $2^{\pf n}$ is the probability that a pair $(\vec{e_1}, \vec{e_2})$, drawn uniformly at random from $D_1 \times D_2$, has $(\vec{e_1} + \vec{e_2}) \in D$.
\end{itemize}
In particular, the occurrence of collisions in $L^f$ is a negligible phenomenon, unless $\ell^f$ approaches ($\log_2 |D| / n) - c$, which is the maximum number of vectors in $D$ with constraint $c$. For a given random knapsack problem, with high probability, the size of any list built by sampling, merging and filtering remains very close to its average (by a Chernoff bound and a union bound on all lists).

Here, $\pf$ depends only on $D_1, D_2$ and $D$. Working with this \emph{filtering probability} is especially useful for writing down our algorithm in Section~\ref{section:hgj-asymmetric}. We give its formula for $\zo$ representations below. Two similar results for $\zom$ and $\zomd$ can be found in \ifeprint Appendix\else the full version of the paper~\cite{DBLP:journals/iacr/BonnetainBSS20}\fi.


\begin{lemma}[Filtering HGJ-style representations]
\label{lemma:hgj-filter}
Let $\vec{e_1} \in D^n[0, \alpha]$ and $\vec{e_2} \in D^n[0, \beta]$ be drawn uniformly at random. The probability that $\vec{e_1} + \vec{e_2} \in D^n[0, \alpha + \beta]$ is $0$ if $\alpha + \beta > 1$, and $2^{\pfilterhgj{\alpha}{\beta}n}$ otherwise, with 
$$\pfilterhgj{\alpha}{\beta} = \bin{1-\alpha}{\beta} - \entr{\beta} = \bin{1-\beta}{\alpha} - \entr{\alpha} \enspace.$$
\end{lemma}

\begin{proof}
The probability that a $\vec{e}_1 + \vec{e}_2$ survives the filtering is:
$$ {n - \alpha n \choose \beta n} / {n \choose \beta n} = {n - \beta n \choose \alpha n} / {n \choose \alpha n} \enspace.$$
Indeed, given a choice of $\alpha n$ bit positions among $n$, the other $\beta n$ bit positions must be compatible, hence chosen among the $(1-\alpha)n$ remaining positions. By taking the $\log_2$, we obtain the formula for the filtering probability. \qed
\end{proof}

\paragraph{Time Complexity of Merging.}
Classically, the time complexity of the merge-and-filter operation is related to the size of the \emph{merged list}.

\begin{lemma}[Classical merging with filtering]
\label{lemma:classical-merging}
Let $L_1$ and $L_2$ be two sorted lists stored in classical memory with random access. 
In $\log_2$, relatively to $n$, and discarding logarithmic factors, merging and filtering $L_1$ and $L_2$ costs a time $\max( \min(\ell_1, \ell_2), \ell_1 + \ell_2 - c )$ and memory $\max( \ell_1, \ell_2, \ell^f)$, assuming that we must store the filtered output list.
\end{lemma}

\begin{proof}Assuming sorted lists, there are two symmetric ways to produce a stream of elements of $L_1 \bowtie_c L_2$: we can go through the elements of $L_1$, and for each one, find the matching elements in $L_2$ by dichotomy search (time $\ell_1 + \max(0, \ell_2 - c)$) or we can exchange the role of $L_1$ and $L_2$. Although we do not need to store $L_1 \bowtie_c L_2$, we need to examine all its elements in order to filter them.\qed
\end{proof}

\subsection{Correctness of the Algorithms}

While the operation of \emph{merging and filtering} is the same as in previous works, our complexity analysis differs~\cite{DBLP:conf/eurocrypt/Howgrave-GrahamJ10,DBLP:conf/eurocrypt/BeckerCJ11,DBLP:conf/pqcrypto/BernsteinJLM13,DBLP:conf/tqc/HelmM18}. We enforce the constraint that the final list contains a single solution, hence if it is of size $2^{n \ell_0}$, we constrain $\ell_0 = 0$. Next, we limit the sizes of the lists so that they do not contain duplicate vectors: these are \emph{saturation constraints}. A list of size $2^{n \ell}$, of vectors sampled from a distribution $D$, with a constraint of $cn$ bits, has the constraint: $\ell \leq \frac{1}{n}\log_2{|D|} - c$. This says that there are not more than $|D| / 2^{cn}$ vectors $\vec{e}$ such that $\vec{e} \cdot \vec{a} = r \pmod{2^{cn}}$ for the (randomly chosen) arbitrary constraint $r$.

Previous works focus on the solution vector $\vec{e}$ and compute the number of \emph{representations} of $\vec{e}$, that is, the number of ways it can be decomposed as a sum: $\vec{e} = \vec{e}_1 + \ldots + \vec{e}_t$ of vectors satisfying the constraints on the distributions. Then, they compare this with the probability that a given representation passes the arbitrary constraints imposed by the algorithm. As their lists contains all the subknapsacks that fulfill the constraint, this really reflects the number of duplicates, and it suffices to enforce that the number of representations is equal to the inverse probability that a representation fulfills the constraint. If the two lists we merge are not of maximal size, the size of the merged list is the number of elements that fulfill the corresponding distribution times the probability that such an element is effectively the sum of two elements in the initial lists.

The two approaches are strictly equivalent, as the probability that the sum of two subknapsacks is valid is exactly the number of representations of the sum, divided by the number of pairs of subknapsacks.

%
%

\subsection{The HGJ Algorithm}

We start our study of classical subset-sum by recalling the algorithm of Howgrave-Graham and Joux~\cite{DBLP:conf/eurocrypt/Howgrave-GrahamJ10}, with the corrected time complexity of~\cite{DBLP:conf/eurocrypt/BeckerCJ11}.
The algorithm builds a merging tree of lists of subknapsacks, with four levels, numbered $3$ down to $0$. Level $j$ contains $2^j$ lists. In total, $8$ lists are merged together into one.
\begin{description}
\item[Level 3.] We build 8 lists denoted $L_0^3 \ldots L_7^3$. They contain \emph{all} subknapsacks of weight $\frac{n}{16}$ on $\frac{n}{2}$ bits, either left or right:
\begin{align*}
\left\{ \begin{matrix}
L_{2i}^3 & = D^{n/2}[0, 1/8] \times \{ 0^{n/2} \} \\
L_{2i+1}^3 & =  \{ 0^{n/2} \} \times D^{n/2}[0, 1/8] 
\end{matrix} \right.
\end{align*}

From Property \ref{property:knapsack_size}, these level-3 lists have size $\ell_3 =  \entr{1/8} / 2$. As the positions set to 1 cannot interfere, these is no filtering when merging $L_{2i}^3$ and $L_{2i+1}^3$.
\item[Level 2.] We merge the lists pairwise with a (random) constraint on $c_2 n$ bits, and obtain 4 filtered lists. The size of the filtered lists plays a role in the memory complexity of the algorithm, but the time complexity depends on the size of the unfiltered lists.

In practice, when we say ``with a constraint on $c_j n$ bits'', we assume that given the subset-sum objective $t$ modulo $2^n$, random values $r_i^j$ such that $\sum_i r_i^j = t \mod{2^{c_jn}}$ are selected at level $j$, and the $r_i^j$ have $c_j n$ bits only. Hence, at this step, we have selected 4 integers on $c_2 n$ bits $r^1_0, r^1_1, r^1_2, r^1_3$ such that $r^1_0 + r^1_1 + r^1_2 + r^1_3 = t \mod{2^{c_2n}}$. The 4 level-2 lists $L_0^2, L_1^2, L_2^2, L_3^2$ have size $\ell_2 = (\entr{1/8} - c_2)$, they contain subknapsacks of weight $\frac{n}{8}$ on $n$ bits.

\begin{remark}
The precise values of these $r_i$ are irrelevant, since they cancel out each other in the end. They are selected at random during a run of the algorithm, and although there could be ``bad'' values of them that affect significantly the computation, this is not expected to happen.
\end{remark}

\item[Level 1.] We merge the lists pairwise with $(c_1 - c_2) n$ new bits of constraint, ensuring that the constraint is compatible with the previous ones. We obtain two filtered lists $L_0^1, L_1^1$, containing subknapsacks of weight $n/4$. They have size:
$$ \ell_1 = 2 \ell_2 - (c_1 - c_2) + \pfilterhgj{1/8}{1/8} $$
where $\pfilterhgj{1/8}{1/8}$ is given by Lemma~\ref{lemma:hgj-filter}.

\item[Level 0.] We find a solution to the subset-sum problem with the complete constraint on $n$ bits. This means that the list $L^0$ must have expected length $\ell_0 = 0$. Note that there remains $(1 - c_1) n$ bits of constraint to satisfy, and the filtering term is similar as before, so:
$$ \ell_0 = 2 \ell_1 - (1-c_1) + \pfilterhgj{1/4}{1/4} \enspace.$$%
\end{description}

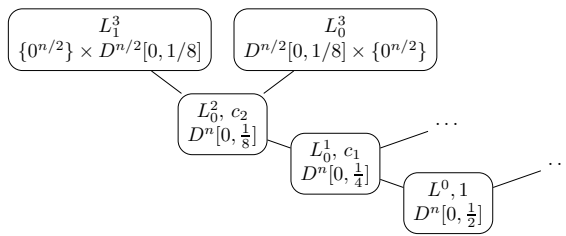
\begin{figure}[htb]
\centering
\scalebox{0.75}{\begin{tikzpicture}[grow=up,nodes={draw,rectangle,rounded corners=.25cm,->}, level 1/.style={sibling distance=40mm, level distance=7mm}, level 2/.style={sibling distance=40mm}, level 3/.style={sibling distance=40mm, level distance=15mm}]
\node{ \begin{tabular}{c} $L^0, 1 $\\ $D^n[0, \frac{1}{2}]$ \end{tabular}}
	child { node[draw=none] {\dots }}
    child { node {\begin{tabular}{c} $L_0^1$, $c_1$ \\ $D^n[0, \frac{1}{4}]$ \end{tabular} }
    	child { node[draw=none] {\dots } }
    	child { node {\begin{tabular}{c} $L_0^2 $, $c_2$ \\ $D^n[0, \frac{1}{8}]$ \end{tabular} }
			child { node {\begin{tabular}{c} $L_0^3 $ \\ $D^{n/2}[0, 1/8] \times \{ 0^{n/2} \}$ \end{tabular} }}
			child { node {\begin{tabular}{c} $L_1^3 $ \\ $\{ 0^{n/2} \} \times D^{n/2}[0, 1/8] $ \end{tabular} }} }
    	};
\end{tikzpicture}}
\caption{The HGJ algorithm (duplicate lists are omitted)}
\label{fig:hgj-original}
\end{figure}

By Lemma~\ref{lemma:classical-merging}, the time complexity of this algorithm is determined by the sizes of the unfiltered lists:~~$ \max\left( \ell_3, 2\ell_3 - c_2, 2\ell_2 - (c_1 - c_2), 2 \ell_1 - (1-c_1) \right)$.

The memory complexity depends of the sizes of the filtered lists: $\max\left( \ell_3, \ell_2, \ell_1\right)$. By a numerical optimization, one obtains a time exponent of $0.337 n$.

\subsection{The BCJ Algorithm and our improvements}\label{subsection:new-classical-bcj}

The HGJ algorithm uses representations to increase artificially the search space. The algorithm of Becker, Coron and Joux~\cite{DBLP:conf/eurocrypt/BeckerCJ11} improves the runtime exponent down to $0.291$ by allowing even more freedom in the representations, which can now contain ``$-1$''s. The ``$-1$''s have to cancel out progressively, to ensure the validity of the final knapsack solution.

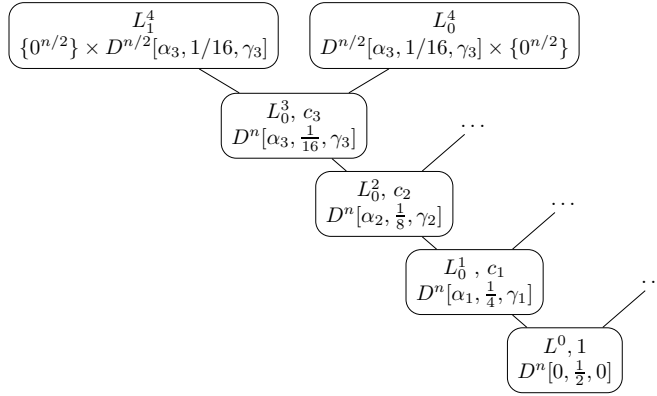
\begin{figure}[htb]
\centering
\scalebox{0.8}{\begin{tikzpicture}[grow=up,nodes={draw,rectangle,rounded corners=.25cm,->}, level 1/.style={sibling distance=30mm, level distance=13mm}, level 2/.style={sibling distance=30mm}, level 3/.style={sibling distance=30mm}, level 4/.style={sibling distance=50mm, level distance=15mm}]
\node{ \begin{tabular}{c} $L^0, 1 $\\ $D^n[0, \frac{1}{2},0]$ \end{tabular}}
    child {node[draw=none] {\dots}}
    child { node {\begin{tabular}{c} $L_0^1$ , $c_1$ \\ $D^n[\alpha_1, \frac{1}{4}, \gamma_1]$ \end{tabular} }
    	child { node[draw=none] {\dots} }
    	child { node {\begin{tabular}{c} $L_0^2 $, $c_2$ \\ $D^n[\alpha_2, \frac{1}{8}, \gamma_2]$ \end{tabular} }
    		child { node[draw=none] {\dots}}
    		child { node {\begin{tabular}{c} $L_0^3 $,  $c_3$ \\ $D^n[\alpha_3, \frac{1}{16}, \gamma_3]$ \end{tabular} }
			child { node {\begin{tabular}{c} $L_0^4 $ \\ $D^{n/2}[\alpha_3, 1/16, \gamma_3] \times \{ 0^{n/2} \}$ \end{tabular} }}
			child { node {\begin{tabular}{c} $L_1^4 $ \\ $\{ 0^{n/2} \} \times D^{n/2}[\alpha_3, 1/16, \gamma_3] $ \end{tabular} }}
	}}};
\end{tikzpicture}}
\caption{Our improved algorithm (duplicate lists are omitted).}
\label{fig:bcj-original}
\end{figure}

We improve over this algorithm in two different ways. First, we relax the constraints $\ell_j + c_j = g(\alpha_j, 1/2^{j+1})$ enforced in~\cite{DBLP:conf/eurocrypt/BeckerCJ11}, as only the inequalities $\ell_j + c_j \leqslant g(\alpha_j, 1/2^{j+1})$ are necessary: they make sure the lists are not larger than the number of distinct elements they can contain. This idea was also implicitly used in~\cite{DBLP:conf/sacrypt/BricoutCDL19}, in the context of syndrome decoding. When optimizing the parameters under these new constraints, we bring the asymptotic time exponent down to $0.289n$. 

\paragraph{$\zomd$ representations.} Next, we allow the value ``2'' in the subknapsacks. This allows us to have more representations for the final solution from the same initial distributions. Indeed, in BCJ, if on a bit the solution is the sum of a ``-1'' and two ``1''s, then it can only pass the merging steps if we first have the ``-1'' that cancels a ``1'', and then the addition of the second ``1''. When allowing ``2''s, we can have the sum of the two ``1's and then at a later step the addition of a ``-1''. The algorithm builds a merging tree with five levels, numbered $4$ down to $0$. Level $j$ contains $2^j$ lists. In total, $16$ lists are merged together into one.
\begin{description}
\item[Level 4.] We build 16 lists $L_0^4 \ldots L_{15}^4$. They contain complete distributions on $\frac{n}{2}$ bits, either left or right, with $\frac{n}{32} + \frac{\alpha_3 n}{2} - \gamma_3 n$ ``1'',  $\frac{\alpha_3 n}{2}$ ``-1'' and $\frac{\gamma_3 n}{2}$ ``$2$'':
\begin{align*}
\left\{ \begin{matrix}
L_{2i}^4 & = D^{n/2}[\alpha_3, 1/16, \gamma_3] \times \{ 0^{n/2} \} \\
L_{2i+1}^4 & =  \{ 0^{n/2} \} \times D^{n/2}[\alpha_3, 1/16, \gamma_3] 
\end{matrix} \right.
\end{align*}
As before, this avoids filtering at the first level. These lists have size: $\ell_4 = f(\alpha_3, 1/16 + \alpha_3 - 2\gamma_3, \gamma_3) / 2$.

\item[Level 3.] We merge into 8 lists $L_0^3 \ldots L_7^3$, with a constraint on $c_3$ bits. As there is no filtering, these lists have size: $\ell_3 = f(\alpha_3, 1/16 + \alpha_3 - 2\gamma_3, \gamma_3) - c_3$.

\item[Level 2.] We now merge and filter. We force a target distribution $D^{n}[\alpha_2, 1/8, \gamma_2]$, with $\alpha_2$ and $\gamma_2$ to be optimized later. There is a first filtering probability $p_2$.
We have $\ell_2 = 2\ell_3 - (c_2 - c_3) + p_2$.
\item[Level 1.] Similarly, we have: $\ell_1 = 2\ell_2 - (c_1 - c_2) + p_1$.
\item[Level 0.] We have $\ell_0 = 2\ell_1 - (1-c_1) + p_0 = 0$, since the goal is to obtain one solution in the list $L_0$.
\end{description}

With these constraints, we find a time $\bigOt{2^{0.2830n}}$ (rounded upwards) with the following parameters:
\begin{align*}
&\alpha_1 = 0.0340, \alpha_2 = 0.0311, \alpha_3 = 0.0202, \gamma_1 = 0.0041, \gamma_2 = 0.0006, \gamma_3 = 0.0001 \\
&c_1 = 0.8067, c_2 = 0.5509, c_3 = 0.2680, p_0 = -0.2829, p_1 = -0.0447, p_2 = -0.0135 \\
&\ell_1 = 0.2382, \ell_2 = 0.2694, \ell_3 = 0.2829, \ell_4 = 0.2755 
\end{align*}


\begin{remark}[On numeric optimizations]
All algorithms since HGJ, including quantum ones, rely on (nonlinear) numeric optimizations. Their correctness is easy to check, since the obtained parameters satisfy the constraints, but there is no formal proof that the parameters are indeed optimal for a given constraint set. The same goes for all algorithms studied in this paper. In order to gain confidence in our results, we tried many different starting points and several equivalent rewriting of the constraints.
\end{remark}


\begin{remark}[Adding more symbols]
In general, adding more symbols (``-2''s, ``3''s, \emph{etc.}) can only increase the parameter search space and improve the optimal time complexity. However, we expect that the improvements from adding more symbols will become smaller and smaller, while the obtained constraints will become more difficult to write down and the parameters harder to optimize. Note that adding ``-1''s decreases the time complexity exponent by $0.048$, while adding ``2''s decreases it only by $0.006$.
\end{remark}

\begin{remark}[On the number of levels]
Algorithms based on merging-and-filtering, classical and quantum, have a number of levels (say, 4 or 5) which must be selected before writing down the constraints. The time complexity is a decreasing function of the number of levels, which quickly reaches a minimum. In all algorithms studied in this paper, adding one more level does not change the cost of the upper levels, which will remain the most expensive. 
\end{remark}

\section{Quantum Preliminaries and Previous Work}
\label{section:quantum-prelim}

In this section, we recall some preliminaries of quantum computation (quantum search and quantum walks) that will be useful throughout the rest of this paper. We also recall previous quantum algorithms for subset-sum. As we consider all our algorithms from the point of view of asymptotic complexities, and neglect polynomial factors in $n$, a high-level overview is often enough, and we will use quantum building blocks as black boxes. The interested reader may find more details in~\cite{nielsen2002quantum}.


\subsection{Quantum Preliminaries}



All the quantum algorithms considered in this paper run in the quantum circuit model, with quantum random-access memory, often denoted as qRAM. ``Baseline'' quantum circuits are simply built using a universal gate set. Many quantum algorithms use qRAM access, and require the circuit model to be augmented with the so-called ``qRAM gate''. This includes subset-sum, lattice sieving and generic decoding algorithms that obtain time speedups with respect to their classical counterparts. Given an input register $1 \leq i \leq r$, which represents the index of a memory cell, and many quantum registers $\ket{x_1 , \ldots x_r}$, which represent stored data, the qRAM gate fetches the data from register $x_i$:
$$ \ket{i} \ket{x_1 , \ldots x_r} \ket{y} \mapsto \ket{i} \ket{x_1 , \ldots x_r} \ket{y \oplus x_i} \enspace.$$


We will use the terminology of~\cite{DBLP:conf/tqc/Kuperberg13} for the qRAM gate:
\begin{itemize}
\item If the input $i$ is classical, then this is the plain quantum circuit model (with classical RAM);
\item If the $x_j$ are classical, we have \emph{quantum-accessible classical memory} (QRACM)
\item In general, we have \emph{quantum-accessible quantum memory} (QRAQM)
\end{itemize}

All known quantum algorithms for subset-sum with a quantum time speedup over the best classical one require QRAQM. For comparison, speedups on heuristic lattice sieving algorithms exist in the QRACM model~\cite{DBLP:journals/dcc/LaarhovenMP15,DBLP:conf/asiacrypt/KirshanovaMPM19}, including the best one to date~\cite{laarhoven2015search}. While no physical architecture for quantum random access has been proposed that would indeed produce a constant or negligible overhead in time, some authors~\cite{DBLP:conf/tqc/Kuperberg13} consider the separation meaningful. If we assign a cost $\bigO{N}$ to a QRACM query of $N$ cells, then we can replace it by \emph{classical} memory. Subset-sum algorithms were studied in this setting by Helm and May~\cite{DBLP:conf/pqcrypto/Helm020}.




\subsubsection{Quantum Search.}

One of the most well-known quantum algorithms is Grover's unstructured search algorithm~\cite{DBLP:conf/stoc/Grover96}. We present here its generalization, amplitude amplification~\cite{brassard2002quantum}.

\begin{lemma}[Amplitude amplification, from \cite{brassard2002quantum}]\label{lemma:grover}
Let $\mathcal{A}$ be a reversible quantum circuit, $f$ a computable boolean function over the output of $\mathcal{A}$, $O_f$ its implementation as a quantum circuit, and $a$ be the initial success probability of $\mathcal{A}$, that is, the probability that $O_f\mathcal{A}\ket0$ outputs ``true''. There exists a quantum reversible algorithm that calls $\bigO{\sqrt{1/a}}$ times $\mathcal{A}$, $\mathcal{A}^\dagger$ and $O_f$, uses as many qubits as $\mathcal{A}$ and $O_f$, and produces an output that passes the test $f$ with probability greater than $\max(a,1-a)$.
\end{lemma}

This is known to be optimal when the functions are black-box oracles~\cite{DBLP:journals/siamcomp/BennettBBV97}.

As we will use quantum search as a subprocedure, we make some remarks similar to~\cite[Appendix A.2]{DBLP:journals/iacr/Naya-PlasenciaS19} and~\cite[Section 5.2]{DBLP:journals/tosc/BonnetainNS19} to justify that, up to additional polynomial factors in time, we can consider it runs with no errors and allows to return all the solutions efficiently.

\begin{remark}[Error in a sequence of quantum searches]
Throughout this paper, we will assume that a quantum search in a search space of size $S$ with $T$ solutions runs in exact time $\sqrt{S/T}$. In practice, there is a constant overhead, but since $S$ and $T$ are always exponential in $n$, the difference is negligible. Furthermore, this is a probabilistic procedure, and it will return a wrong result with a probability of the order $\sqrt{T/S}$. As we can test if an error occurs, we can make it negligible by redoing the quantum search polynomially many times.
\end{remark}

\begin{remark}[Finding all solutions]\label{remark:coupon-coll}
Quantum search returns a solution among the $T$ possibilities, selected uniformly at random. Finding all solutions is then an instance of the coupon collector problem with $T$ coupons~\cite{newman1960double}; all coupons are collected after on average $\bigO{T\log(T)}$ trials. However, \emph{in the QRACM model}, which is assumed in this paper, this logarithmic factor disappears. We can run the search of Lemma~\ref{lemma:grover} with a new test function that returns $0$ if the output of $\mathcal{A}$ is incorrect, \emph{or} if it is correct but has already been found. The change to the runtime is negligible, and thus, we collect all solutions with only $\bigO{T}$ searches.
\end{remark}

\subsubsection{Quantum Walks.}\label{subsection:mnrs}


Quantum walks can be seen as a generalization of quantum search. They allow to obtain polynomial speedups on many unstructured problems, with sometimes optimal results (\emph{e.g.} Ambainis' algorithm for element distinctness~\cite{DBLP:journals/siamcomp/Ambainis07}). In this paper, we consider walks in the MNRS framework~\cite{mnrs11}.

Let $G = (V,E)$ be an undirected, connected, regular graph, such that some vertices of $G$ are ``marked''. Let $\epsilon$ be the fraction of marked vertices, that is, a random vertex has a probability $\epsilon$ of being marked. Let $\delta$ be the spectral gap of $G$, which is defined as the difference between its two largest eigenvalues.

In a \emph{classical} random walk on $G$, we can start from any vertex and reach the stationary distribution in approximately $\frac{1}{\delta}$ random walk steps. Then, such a random vertex is marked with probability $\epsilon$. Assume that we have a procedure \texttt{Setup} that samples a random vertex to start with in time $\mathsf{S}$, \texttt{Check} that verifies if a vertex is marked or not in time $\mathsf{C}$ and \texttt{Update} that performs a walk step in time $\mathsf{U}$, then we will have found a marked vertex in expected time:  $ \mathsf{S} + \frac{1}{\epsilon} \left( \frac{1}{\delta} \mathsf{U} + \mathsf{C} \right) \enspace. $

Quantum walks reproduce the same process, except that their internal state is not a vertex of $G$, but a superposition of vertices. The walk starts in the uniform superposition $\sum_{v \in V} \ket{v}$, which must be generated by the \texttt{Setup} procedure. It repeats $\sqrt{1/ \epsilon}$ iterations that, similarly to amplitude amplification, move the amplitude towards the marked vertices. An update produces, from a vertex, the superposition of its neighbors. Each iteration does not need to repeat $\frac{1}{\delta}$ vertex updates and, instead, takes a time equivalent to $\sqrt{1/\delta}$ updates to achieve a good mixing.
Thanks to the following theorem~, we will only need to specify the setup, checking and update unitaries.

\begin{theorem}[Quantum walk on a graph (adapted from~\cite{mnrs11})]\label{thm:qwgraph}
Let $G=(V,E)$ be a regular graph with spectral gap $\delta>0$. Let $\epsilon>0$ be a lower bound on the probability that a vertex chosen randomly of $G$ is marked. For a random walk on $G$, let $\mathsf{S},\mathsf{U}, \mathsf{C}$ be the setup, update and checking cost. Then there exists a quantum algorithm that with high probability finds a marked vertex in time $$ \bigO{ \mathsf{S} + \frac{1}{\sqrt{\epsilon}} \left( \frac{1}{\sqrt{\delta}}\mathsf{U} + \mathsf{C} \right)} .$$
\end{theorem}


\subsection{Solving Subset-sum with Quantum Walks}\label{subsection:subsumqw}

In 2013, Bernstein, Jeffery, Lange and Meurer~\cite{DBLP:conf/pqcrypto/BernsteinJLM13} constructed quantum subset sum algorithms inspired by Schroeppel-Shamir~\cite{DBLP:journals/siamcomp/SchroeppelS81} and HGJ~\cite{DBLP:conf/eurocrypt/Howgrave-GrahamJ10}. We briefly explain the idea of their quantum walk for HGJ. The graph $G$ that they consider is a product Johnson graph. We recall formal definitions from~\cite{DBLP:conf/pqcrypto/KachigarT17}.

\begin{definition}[Johnson graph]
A Johnson graph $J(N,R)$ is an undirected graph whose vertices are the subsets of $R$ elements among a set of size $N$, and there is an edge between two vertices $S$ and $S'$ iff $|S \cap S'| = R-1$, in other words, if $S'$ can be obtained from $S$ by replacing an element. Its spectral gap is given by $\delta = \frac{N}{R(N-R)}$.
\end{definition}

\begin{theorem}[Cartesian product of Johnson graphs~\cite{DBLP:conf/pqcrypto/KachigarT17}]
Let $J^m(N, R)$ be defined as the cartesian product of $m$ Johnson graphs $J(N,R)$, \emph{i.e.}, a vertex in $J^m(N, R)$ is a tuple of $m$ subsets $S_1, \ldots S_m$ and there is an edge between $S_1, \ldots S_m$ and $S_1', \ldots S_m'$ iff all subsets are equal at all indices except one index $i$, which satisfies $|S_i \cap S_i'| = R-1$. Then it has ${N \choose R}^{m}$ vertices and its spectral gap is greater than $\frac{1}{m} \frac{N}{R(N-R)}$.
\end{theorem}

In ~\cite{DBLP:conf/pqcrypto/BernsteinJLM13}, a vertex contains a product of 8 sublists $L_0'^3 \subset L_0^3, \ldots, L_7'^3 \subset L_7^3$ of a smaller size than the classical lists: $\ell < \ell_3$. There is an edge between two vertices if we can transform one into the other by replacing only one element in one of the sublists. The spectral gap of such a graph is (in $\log_2$, relative to $n$) $- \ell$.

In addition, each vertex has an internal data structure which reproduces the HGJ merging tree, from level 3 to level 0. Since the initial lists are smaller, the list $L^0$ is now of expected size $8(\ell - \ell_3)$ (in $\log_2$, relative to $n$), \emph{i.e.} the walk needs to run for $4(\ell_3 - \ell)$ steps. Each step requires $\ell / 2$ updates.

In the \texttt{Setup} procedure, we simply start from all choices for the sublists and build the tree by merging and filtering. Assuming that the merged lists have decreasing sizes, the setup time is $\ell$. The vertex is marked if it contains a solution at level 0. Hence, checking if a vertex is marked takes time $\mathsf{C} = 1$, but the update procedure needs to ensure the consistency of the data structure. Indeed, when updating, we remove an element $\vec{e}$ from one of the lists $L_i'^3$ and replace it by a $\vec{e'}$ from $L_i^3$. We then have to track all subknapsacks in the upper levels where $\vec{e}$ intervened, to remove them, and to add the new collisions where $\vec{e'}$ intervenes.

Assuming that the update can run in $\poly(n)$, an optimization with the new parameter $\ell$ yields an exponent 0.241. In~\cite{DBLP:conf/pqcrypto/BernsteinJLM13}, the parameters are such that on average, a subknapsack intervenes only in a single sum at the next level. The authors propose to simply limit the number of elements to be updated at each level, in order to guarantee a constant update time.

\paragraph{Quantum Walk Based on BCJ.}
In~\cite{DBLP:conf/tqc/HelmM18}, Helm and May quantize, in the same way, the BCJ algorithm. They add ``-1'' symbols and a new level in the merging tree data structure, reaching a time exponent of 0.226. But they remark that this result depends on a conjecture, or a heuristic, that was implicit in~\cite{DBLP:conf/pqcrypto/BernsteinJLM13}.

\begin{heuristic}[Helm-May]\label{heuristic:helm-may}
In these quantum walk subset-sum algorithms, an update with expected constant time $\mathsf{U}$ can be replaced by an update with exact time $\mathsf{U}$ without affecting the runtime of the algorithm, up to a polynomial factor.
\end{heuristic}

Indeed, it is easy to construct ``bad'' vertices and edges for which an exact update, \emph{i.e.} the complete reconstruction of the merging tree, will take exponential time: by adding a single new subknapsack $\vec{e}$, we find an exponential number of pairs $\vec{e} + \vec{e'}$ to include at the next level. So we would like to update only a few elements among them. But in the MNRS framework, the data structure of a vertex must depend solely on the vertex itself (\emph{i.e.} on the lowest-level lists in the merging tree). And if we do as proposed in~\cite{DBLP:conf/pqcrypto/BernsteinJLM13}, we add a dependency on the path that lead to the vertex, and lose the consistency of the walk.

In a related context, the problem of ``quantum search with variable times'' was studied by Ambainis~\cite{DBLP:journals/mst/Ambainis10}. In a quantum search for some $x$ such that $f(x) = 1$, in a set of size $N$, if the time to evaluate $f$ on $x$ is always $1$, then the search requires time $\bigO{\sqrt{N}}$. Ambainis showed that if the elements have different evaluation times $t_1, \ldots t_N$, then the search now requires $\widetilde{\mathcal{O}}(\sqrt{ t_1^2 + \ldots + t_N^2 })$, the geometric mean of $t_1, \ldots t_N$. As quantum search can be seen as a particular type of quantum walk, this shows that Heuristic~\ref{heuristic:helm-may} is wrong in general, as we can artificially create a gap between the geometric mean and expectation of the update time $\mathsf{U}$; but also, that it may be difficult to actually overcome. In this paper, we will obtain different heuristic and non-heuristic times.

\section{Quantum Asymmetric HGJ}
\label{section:hgj-asymmetric}

In this section, we give the first quantum algorithm for the subset-sum problem, \emph{in the QRACM model}, with an asymptotic complexity smaller than BCJ.

\subsection{Quantum Match-and-Filter}

We open this section with some technical lemmas that replace the classical merge-and-filter Lemma~\ref{lemma:classical-merging}. In this section, we will consider a merging tree as in the HGJ algorithm, but this tree will be built using quantum search. The following lemmas bound the expected time of merge-and-filter \emph{and} match-and-filter operations performed quantumly, in the QRACM model. This will have consequences both in this section and in the next one.


First, we remark that we can use a much more simple data structure than the ones in~\cite{DBLP:journals/siamcomp/Ambainis07,DBLP:conf/pqcrypto/BernsteinJLM13}. In this data structure, we store pairs $\vec{e}, \vec{e} \cdot \vec{a}$ indexed by $\vec{e} \cdot \vec{a} \mod M$ for some $M \simeq 2^m$. 

\begin{definition}[Unique modulus list]\label{def:unique-mod}
A \emph{unique modulus list} is a qRAM data structure $\mathcal{L}(M)$ that stores at most $M$ entries $(\vec{e}, \vec{e} \cdot \vec{a})$, indexed by $\vec{e} \cdot \vec{a} \mod M$, and supports the following operations:
\begin{itemize}
\item Insertion: inserts the entry $(\vec{e}, \vec{e} \cdot \vec{a})$ if the modulus is not already occupied;
\item Deletion: deletes $(\vec{e}, \vec{e} \cdot \vec{a})$ (not necessary in this section)
\item Query in superposition: returns the superposition of all entries $(\vec{e}, \vec{e} \cdot \vec{a})$ with some modular condition on $\vec{e} \cdot \vec{a}$, \emph{e.g.} $\vec{e} \cdot \vec{a} = t \mod M'$ for some $t$ and some modulus $M'$.
\end{itemize}
\end{definition}

Note that all of these operations, \emph{including} the query in superposition of all the entries with a given modulus, cost $\bigO{1}$ qRAM gates only. For the latter, we need only some Hadamard gates to prepare the adequate superposition of indices. Furthermore, the list remains sorted by design.

Next, we write a lemma for quantum \emph{matching} with filtering, in which one of the lists is not written down. We start from a unitary that produces the uniform superposition of the elements of a list $L_1$, and we wrap it into an amplitude amplification, in order to obtain a unitary that produces the uniform superposition of the elements of the merged-and-filtered list.

\begin{lemma}[Quantum matching with filtering]\label{lemma:quantum-matching-filtering}
Let $L_2$ be a list stored in QRACM (with the \emph{unique modulus list} data structure of Definition~\ref{def:unique-mod}). Assume given a unitary $U$ that produces in time $t_{L_1}$ the uniform superposition of $L_1 = x_0, \ldots x_{2^m -1}$ where $x_i = (\vec{e_i}, \vec{e_i} \cdot \vec{a})$.
We merge $L_1$ and $L_2$ with a modular condition of $cn$ bits and a filtering probability $p$. Let $L$ be the merged list and $L^f$ the filtered list. Assume $|L^f| \geq 1$. Then there exists a unitary $U'$ producing the uniform superposition of $L^f$ in time: $\bigO{ \frac{t_{L_1}}{\sqrt{p}} \max( \sqrt{2^{cn} / |L_2|} , 1 ) }$.
\end{lemma}

Notice that this is also the time complexity to produce a single random element of $L^f$. If we want to produce and store the whole list $L^f$, it suffices to multiply this complexity by the number of elements in $L^f$ (\emph{i.e.} $p|L_1||L_2|/2^{cn}$). We would obtain: $\bigO{ t_{L_1} \sqrt{p} \max \left( |L_1| \sqrt{\frac{|L_2|}{2^{cn}}}  , \frac{|L_1| |L_2|}{2^{cn}} \right) } \enspace.$

\begin{proof}
Since $L_2$ is stored in a unique modulus list, all its elements have distinct moduli. Note that the \emph{expected} sizes of $L$ and $L^f$ follow from Heuristic~\ref{heur1}. Although the number of iterations of quantum search should depend on the \emph{real} sizes of these lists, the concentration around the average is so high (given by Chernoff bounds) that the error remains negligible if we run the search with the expected number of iterations. We separate three cases.
\begin{itemize}
\item If $|L_2| < 2^{cn}$, then we have no choice but to make a quantum search on elements of $L_1$ that match the modular constraint and pass the filtering step, in time: $\bigO{ t_{L_1} \sqrt{ \frac{2^{cn}}{L_2 p} } }$.

\item If $|L_2| > 2^{cn}$ but $|L_2| < 2^{cn} / p$, an element of $L_1$ will always pass the modular constraint, with more than one candidate, but in general all these candidates will be filtered out. Given an element of $L_1$, producing the superposition of these candidates is done in time $1$, so finding the one that passes the filter, if there is one, takes time $\sqrt{|L_2|/2^{cn}}$. Next, we wrap this in a quantum search to find the ``good'' elements of $L_1$ (passing the two conditions), with $ \bigO{\sqrt{ 2^{cn} / p L_2}}$ iterations. The total time is:
$$\bigO{ \sqrt{ \frac{2^{cn}}{L_2 p} } \times \left(  \sqrt{|L_2|/2^{cn}} \times t_{L_1} \right) =  \frac{t_{L_1}}{\sqrt{p}} }\enspace. $$

\item If $|L_2| > 2^{cn}/ p$, an element of $L_1$ yields on average more than one filtered candidate. Producing the superposition of the modular candidates is done in time $\bigO{1}$ thanks to the data structure, then finding the superposition of filtered candidates requires $1/\sqrt{p}$ iterations. The total time is: $\bigO{ t_{L_1} / \sqrt{p} }$.
\end{itemize}
The total time in all cases is: $\bigO{ \frac{t_{L_1}}{\sqrt{p}} \max( \sqrt{2^{cn} / |L_2|} , 1 ) }$. Note that classically, the coupon collector problem would have added a polynomial factor, but this is not the case here thanks to QRACM (Remark~\ref{remark:coupon-coll}). \qed
\end{proof}


In the QRACM model, we have the following corollary for merging and filtering two lists of equal size. This result will be helpful in Section~\ref{subsection:hgj-quantum-second} and~\ref{section:new-qw}.

\begin{corollary}\label{lemma:symmetric-merging}
Consider two lists $L_1, L_2$ of size $|L_1| = |L_2| = |L|$ exponential in $n$. We merge $L_1$ and $L_2$ with a modular condition of $cn$ bits, and filter with a probability $p$. Assume that $2^{cn} < |L|$. Then $L^f$ can be written down in quantum time: $\bigO{ \sqrt{p} \frac{|L|^2}{2^{cn}} } $.
\end{corollary}

\begin{proof}
We do a quantum search to find each element of $L^f$. We have $t_{L_1} = \bigO{1}$ since it is a mere QRACM query, and we use Lemma~\ref{lemma:quantum-matching-filtering}. \qed
\end{proof}

\subsection{Revisiting HGJ}
\label{subsection:hgj-quantum-first}

We now introduce our new algorithm for subset-sum in the QRACM model.

Our starting point is the HGJ algorithm. Similarly to~\cite{DBLP:journals/iacr/Naya-PlasenciaS19}, we use a merging tree in which the lists at a given level may have different sizes. Classically, this does not improve the time complexity. However, quantumly, we will use quantum filtering. Since our algorithm does not require to write data in superposition, only to read from classical registers with quantum random access, we require only QRACM instead of QRAQM.

In the following, we consider that all lists, except $L_0^3, L_0^2, L_0^1, L^0$, are built with classical merges. The final list $L^0$, containing (expectedly) a single element, and a branch leading to it, are part of a nested quantum search. Each list $L_0^3, L_0^2, L_0^1, L^0$ corresponds either to a search space, the solutions of a search, or both. We represent this situation on Fig.~\ref{fig:hgj-quantum}. Our procedure runs as follows:
\begin{enumerate}
\item (Classical step): build the \emph{intermediate lists} $L_1^3, L_1^2, L_1^1$ and store them using a \emph{unique modulus list} data structure (Definition~\ref{def:unique-mod}).
\item (Quantum step): do a quantum search on $L_0^3$. To test a vector $\vec{e} \in L_0^3$:
\begin{itemize}
\item Find $\vec{e}_3 \in L_1^3$ such that $\vec{e} + \vec{e}_3$ passes the $c^2_0n$-bit modular constraint (assume that there is at most one such solution). There is no filtering here.
\item Find $\vec{e}_2 \in L_1^2$ such that $(\vec{e} + \vec{e}_3) + \vec{e}_2$ passes the additional $(c^1-c^2_0)n$-bit constraint.
\item If it also passes the filtering step, find $\vec{e}_1 \in L_1^1$ such that $(\vec{e} + \vec{e}_3 + \vec{e}_2) + \vec{e}_1$ is a solution to the knapsack problem (and passes the filter).
\end{itemize}
\end{enumerate}

Structural constraints are imposed on the tree, in order to guarantee that there exists a knapsack solution. The only difference between the quantum and classical settings is in the optimization goal: the final time complexity.

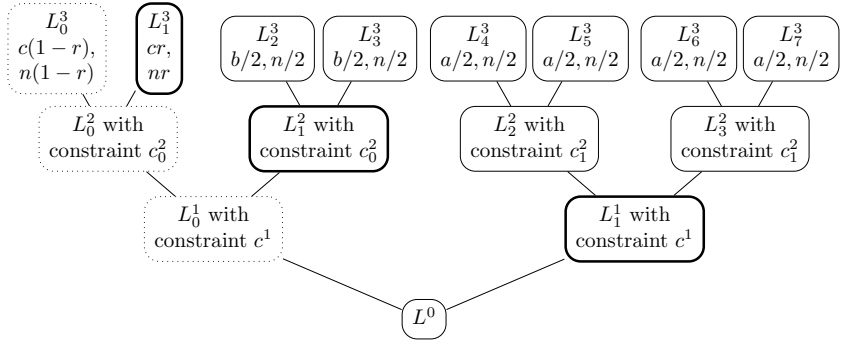
\begin{figure}[htb]
\centering
\scalebox{0.8}{\begin{tikzpicture}[grow=up,nodes={draw,rectangle,rounded corners=.25cm,->}, level 1/.style={sibling distance=70mm}, level 2/.style={sibling distance=35mm}, level 3/.style={sibling distance=17mm}]
\node{ \begin{tabular}{c} $L^0$ \end{tabular}}
    child { node[very thick] {\begin{tabular}{c} $L_1^1$ with \\ constraint $c^1$ \end{tabular} }
    	child { node {\begin{tabular}{c} $L_3^2 $ with \\ constraint $c_1^2$ \end{tabular} }
    		child { node {\begin{tabular}{c} $L_7^3 $ \\ $a/2, n/2$\end{tabular} }}
			child { node {\begin{tabular}{c} $L_6^3 $ \\ $a/2, n/2$\end{tabular} }}}
    	child { node {\begin{tabular}{c} $L_2^2 $ with \\ constraint $c_1^2$ \end{tabular} }
			child { node {\begin{tabular}{c} $L_5^3 $ \\ $a/2, n/2$\end{tabular} }}
			child { node {\begin{tabular}{c} $L_4^3 $ \\ $a/2, n/2$\end{tabular} }} }}
    child { node[dotted] {\begin{tabular}{c} $L_0^1$ with \\ constraint $c^1$ \end{tabular} }
    	child { node[very thick] {\begin{tabular}{c} $L_1^2 $ with \\ constraint $c_0^2$ \end{tabular} }
    		child { node {\begin{tabular}{c} $L_3^3 $\\ $b/2, n/2$\end{tabular} }}
			child { node {\begin{tabular}{c} $L_2^3 $\\ $b/2, n/2$ \end{tabular} }} 
			 }
    	child { node[dotted] {\begin{tabular}{c} $L_0^2 $ with \\ constraint $c_0^2$ \end{tabular} }
			child { node[very thick]         {\begin{tabular}{c} $L_1^3 $\\ $cr$, \\ $n r$ \end{tabular} }}			
			child { node[dotted] {\begin{tabular}{c} $L_0^3 $\\ $c(1-r)$, \\ $n(1-r)$ \end{tabular} }} 
			 }};
\end{tikzpicture}}
\caption{Quantum HGJ algorithm. Dotted lists are search spaces (they are not stored). Bold lists are stored in QRACM. In Section~\ref{subsection:hgj-quantum-second}, $L_2^2$ and $L_3^2$ are also stored in QRACM.}
\label{fig:hgj-quantum}
\end{figure}

\paragraph{Structural Constraints.}
We now introduce the variables and the structural constraints that determine the shape of the tree in Fig.~\ref{fig:hgj-quantum}. The asymmetry happens both in the weights at level 0 and at the constraints at level 1 and 2. We write $\ell_i^j = (\log_2 | L_i^j|) / n$. With the lists built classically, we expect a symmetry to be respected, so we have: $\ell_2^3 = \ell_3^3$, $\ell_4^3 = \ell_5^3 = \ell_6^3 = \ell_7^3$, $\ell_2^2 = \ell_3^2$. We also tweak the left-right split at level 0: lists from $L_2^3$ to $L_7^3$ have a standard balanced left-right split; however, we introduce a parameter $r$ that determines the proportion of positions set to zero in list $L_0^3$: in $L_0^3$, the vectors weigh $cn(1-r)$ on a support of size $n(1-r)$, instead of $cn/2$ on a support of size $n/2$. In total we have $c + b + 2a = \frac{1}{2}$, as the weight of the solution is supposed to be exactly $n/2$.

Then we note that:
\begin{itemize}
\item The lists at level 3 have a maximal size depending on the corresponding weight of their vectors:
$$ \ell_0^3 \leq \entr{c} (1-r),~~ \ell_1^3 \leq \entr{c} r,~~ \ell_2^3 = \ell_3^3 \leq \entr{b}/2,~~ \ell_4^3 \leq \entr{a}/2 $$
\item The lists at level 2 cannot contain more representations than the filtered list of all subknapsacks of corresponding weight:
$$ \ell_0^2 \leq \entr{c} - c_0^2,~~ \ell_1^2 \leq \entr{b} - c_0^2,~~ \ell_2^2 = \ell_3^2 \leq \entr{a} - c_1^2 $$
\item Same at levels 1 and 0: ~~ $ \ell_0^1 \leq \entr{c + b} - c^1,~~ \ell_1^1 \leq \entr{2a} - c^1 $
\item The merging at level 2 is exact (there is no filtering):
$$ \ell_0^2 = \ell_0^3 + \ell_1^3 - c_0^2,~~ \ell_1^2 = \ell_2^3 + \ell_3^3 - c_0^2,~~ \ell_2^2 = \ell_4^3 + \ell_5^3 - c_1^2, ~~\ell_3^2 = \ell_6^3 + \ell_7^3 - c_1^2, $$
\item At level 1, with a constraint $c^1 \geq c_0^2, c_1^2$ that subsumes the previous ones:
$$ \ell_0^1 = \ell_0^2 + \ell_1^2 - c^1 + c_0^2 + \pfilterhgj{b}{c}~,~~~~ \ell_1^1 = \ell_2^2 + \ell_3^2 - c^1 + c_1^2 + \pfilterhgj{a}{a}$$
\item And finally at level 0: ~~ $ \ell^0 = 0 = \ell_0^1 + \ell_1^1 - (1-c^1) + \pfilterhgj{b +c}{2a} $
\end{itemize}

\paragraph{Classical Optimization.}
All the previous constraints depend on the problem, not on the computation model. Now we can get to the time complexity in the classical setting, that we want to minimize:
\begin{multline*}
\max \big( \ell_4^3, \ell_2^3, \ell_1^3, \ell_2^3 + \ell_3^3 - c_0^2, \ell_4^3 + \ell_5^3 - c_1^2, \ell_2^2 + \ell_3^2 - c^1 + c_1^2, \\
\ell_0^3 + \max( \ell_1^3 - c_0^2, 0 ) + \max( \ell_1^2 - c^1 + c_0^2 ,0 )+\max( \pfilterhgj{b}{c}+\ell^1_1-(1-c^1), 0 ) \big) \enspace.
\end{multline*}
The last term corresponds to the exhaustive search on $\ell_0^3$. In order to keep the same freedom as before, it is possible that an element of $L_0^3$ matches against \emph{several} elements of $L_1^3$, all of which yield a potential solution that has to be matched against $L_1^2$, \emph{etc.} Hence for each element of $L_0^3$, we find the expected $\max( \ell_1^3 - c_0^2, 0 )$ candidates matching the constraint $c_0^1$. For each of these candidates, we find the expected $\max( \ell_1^2 - c^1 + c_0^2 ,0 )$ candidates matching the constraint $c^1$. For each of these candidates, if it passes the filter, we search for a collision in $L_1^1$; this explains the $\max( \pfilterhgj{b}{c}+\ell^1_1-(1-c^1), 0 ) $ term. In the end, we check if the final candidates pass the filter on the last level. 

We verified that optimizing the classical time under our constraints gives the time complexity of HGJ.

\paragraph{Quantum Optimization.}
The time complexity for producing the intermediate lists is unchanged. The only difference is the way we find the element in $L_0^3$ that will lead to a solution, which is a nested sequence of quantum searches.
\begin{itemize}
\item We can produce the superposition of all elements in $L_0^2$ in time
$$ t_{2} = \frac{1}{2} \max( c_0^2 - \ell_1^3, 0 ) $$
\item By Lemma~\ref{lemma:quantum-matching-filtering}, we can produce the superposition of all elements in $L_0^1$ in time
$$ t_2 - \frac{1}{2} \pfilterhgj{b}{c} + \frac{1}{2} \max \left( c^1 - c_0^2 - \ell_1^2, 0 \right)  $$
\item Finally, we expect that there are $ \left( \ell_0^2 + \ell_1^2 - c^1 + c_0^2 + \pfilterhgj{b}{c} \right)$ elements in $L_0^1$, which gives the number of iterations of the quantum search.
\end{itemize}
 The time of this search is:
$$ \frac{1}{2} \left( \ell_0^1 +  \max( c_0^2 - \ell_1^3, 0 ) - \pfilterhgj{b}{c} + \max \left( c^1 - c_0^2 - \ell_1^2, 0 \right) \right)  $$
and the total time complexity is:
\begin{multline*}
\max \big( \ell_4^3, \ell_2^3, \ell_1^3, \ell_2^3 + \ell_3^3 - c_0^2, \ell_4^3 + \ell_5^3 - c_1^2, \ell_2^2 + \ell_3^2 - c^1 + c_1^2, \\
\frac{1}{2} \left( \ell_0^1 +  \max( c_0^2 - \ell_1^3, 0 ) - \pfilterhgj{b}{c} + \max \left( c^1 - c_0^2 - \ell_1^2, 0 \right) \right) \big)
\end{multline*}

We obtain a quantum time complexity exponent of $0.2374$ with this method (the detailed parameters are given in Table~\ref{table:quantum-hgj-parameters}).

\subsection{Improvement via Quantum Filtering}
\label{subsection:hgj-quantum-second}

Let us keep the tree structure of Figure~\ref{fig:hgj-quantum} and its structural constraints. The final quantum search step is already made efficient with respect to the filtering of representations, as we only pay half of the filtering term $\pfilterhgj{b}{c}$. However, we can look towards the \emph{intermediate lists} in the tree, \emph{i.e.} $L_1^3, L_1^2, L_1^1$. The merging at the first level is exact: due to the left-right split, there is no filtering of representations, hence the complexity is determined by the size of the output list. However, the construction of $L_1^1$ contains a filtering step. Thus, we can use Corollary~\ref{lemma:symmetric-merging} to produce the elements of $L_1^1$ faster and reduce the time complexity from: $\ell_2^2 + \ell_3^2 - c^1 + c_1^2$ to: $\ell_2^2 + \ell_3^2 - c^1 + c_1^2 + \frac{1}{2} \pfilterhgj{a}{a}$. By optimizing with this time complexity, we obtain a time exponent $0.2356$ (the detailed parameters are given in Table~\ref{table:quantum-hgj-parameters}). The corresponding memory is $0.2356$ (given by the list $L_1^3$).


\begin{table}
\centering
\small
\caption{Optimization results for the quantum asymmetric HGJ algorithm (in $\log_2$ and relative to $n$), rounded to four digits. The time complexity is an upper bound.}
\label{table:quantum-hgj-parameters}

\begin{tabular}{lcccccccccc}
\toprule
Variant &									  Time   & $a$    & $b$    & $c$    & $\ell^3_0$ & $\ell^3_1$ & $\ell^3_2$ & $\ell^3_4$ & $\ell^2_0$ & $\ell^1_0$ \\
\midrule
Classical 									& 0.3370 & 0.1249 & 0.11   & 0.1401 & 0.3026    &  0.2267 & 0.25   & 0.2598 & 0.3369 & 0.3114 \\
Section~\ref{subsection:hgj-quantum-first}  & 0.2374 & 0.0951 & 0.0951 & 0.2146 & 0.4621    &  0.2367 & 0.2267 & 0.2267 & 0.4746 & 0.4395 \\
Section~\ref{subsection:hgj-quantum-second} & 0.2356 & 0.0969 & 0.0952 & 0.2110 & 0.4691    &  0.2356 & 0.2267 & 0.2296 & 0.4695 & 0.4368 \\
\bottomrule
\end{tabular}

\end{table}


\begin{remark}[More improvements]
We have tried increasing the tree depth or changing the tree structure, but it does not seem to bring any improvement. In theory, we could allow for more general representations involving ``-1'' and ``2''. However, computing the filtering probability, when merging two lists of subknapsacks in $D^n[\alpha, \beta, \gamma]$ \emph{with different distributions} becomes much more technical. We managed to compute it for $D^n[\alpha, \beta]$, but the number of parameters was too high for our numerical optimizer, which failed to converge.
\end{remark}

\subsection{Quantum Time-Memory Tradeoff}

In the original HGJ algorithm, the lists at level 3 contain full distributions $D^{n/2}[0, 1/8]$. By reducing their sizes to a smaller exponential, one can still run the merging steps, but the final list $L^0$ is of expected size exponentially small in $n$. Hence, one must redo the tree many times.
This general time-memory tradeoff is outlined in~\cite{DBLP:conf/eurocrypt/Howgrave-GrahamJ10} and is also reminiscent of Schroeppel and Shamir's algorithm~\cite{DBLP:journals/siamcomp/SchroeppelS81}, which can actually be seen as repeating $2^{n/4}$ times a merge of lists of size $2^{n/4}$, that yields $2^{-n/4}$ solutions on average. 



\paragraph{Asymmetric Tradeoff.}
The tradeoff that we propose is adapted to the QRACM model. It consists in increasing the asymmetry of the tree: we reduce the sizes of the intermediate lists $L_1^3, L_1^2, L_1^1$ in order to use less memory; this in turn increases the size of $L^3_0, L^2_0$ and $L^1_0$ in order to ensure that a solution exists. We find that this tradeoff is close to the time-memory product curve $T M = 2^{n/2}$, and actually slightly better (the optimal point when $m = 0.2356$ has $T M = 2^{0.4712 n}$). This is shown on Figure~\ref{fig:tradeoff}. At $m = 0$, we start at $2^{n/2}$, where $L_0^3$ contains all vectors of Hamming weight $n/2$.


\begin{fact}
For any memory constraint $m \leq 0.2356$ (in $\log_2$ and proportion of $n$), the optimal time complexity in the quantum asymmetric HGJ algorithm of Section~\ref{subsection:hgj-quantum-second} is lower than $\bigOt{2^{n/2 - m}}$.
\end{fact}


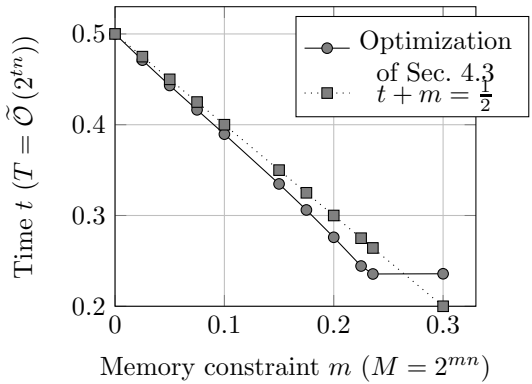
\begin{figure}[htb]
\centering
\begin{tikzpicture}
\begin{axis}[
scale=0.7,
legend pos=outer north east,
xlabel={Memory constraint $m$ ($M = 2^{m n}$)},
ylabel={Time $t$ ($T = \bigOt{2^{t n}}$)},
ymin=0.2,
xmin=0,
xmajorgrids,
ymajorgrids,
legend pos=outer north east,
cycle list name=my black white,
legend style={cells={align=left},name=legend,at={(0.5,0.8)},anchor=west}
]
\addplot coordinates {
(0, 0.5) (0.025, 0.47112909279) (0.05, 0.443214240601) (0.075, 0.416386741403) (0.1, 0.389598290015) (0.15, 0.334770105658) (0.175, 0.306104394719) (0.2, 0.275984845837) (0.225, 0.244335325227) (0.2358, 0.235565150451) (0.3, 0.235795710644)
};
\addlegendentry{\begin{tabular}{c} Optimization \\ of Sec.~\ref{subsection:hgj-quantum-second} \end{tabular}}
\addplot coordinates {
(0, 0.5) (0.025, 0.475) (0.05, 0.45) (0.075, 0.425) (0.1, 0.4) (0.15, 0.35) (0.175, 0.325) (0.2, 0.3) (0.225, 0.275) (0.2358, 0.2642) (0.3, 0.2)
};
\addlegendentry{$t + m = \frac{1}{2}$}
\end{axis}
\end{tikzpicture}
\caption{Quantum time-memory tradeoff of the asymmetric HGJ algorithm}
\label{fig:tradeoff}
\end{figure}
\paragraph{Improving the QRACM usage.}
In trying to reduce the quantum or quantum-accessible hardware used by our algorithm, it makes sense to draw a line between QRACM and classical RAM, \emph{i.e.} between the part of the memory that is actually accessed quantumly, and the memory that is used only classically. We now try to enforce the constraint only on the QRACM, using possibly more RAM. In this context, we cannot produce the list $L_1^1$ via quantum filtering. The memory constraint on lists $L_1^3, L_1^2, L_1^1$ still holds; however, we can increase the size of lists $L_4^3, L_5^3, L_6^3, L_7^3, L_2^2, L_3^2$.

\begin{fact}
For any QRACM constraint $m \leq 0.2356$, the optimal time complexity obtained by using more RAM is always smaller than the best optimization of Section~\ref{subsection:hgj-quantum-second}.
\end{fact}

The difference remains only marginal, as can be seen in Table~\ref{table:hgj-tradeoffs}, but it shows a tradeoff between quantum and classical resources.

\begin{table}[htbp]
\centering
\caption{Time-memory tradeoffs (QRACM) for three variants of our asymmetric HGJ algorithm, obtained by numerical optimization, and rounded upwards. The last variant uses more classical RAM than the QRACM constraint.}
\label{table:hgj-tradeoffs}

\begin{tabular}{ccccccc}
\toprule
QRACM & \multicolumn{2}{c}{ Section~\ref{subsection:hgj-quantum-first} } & \multicolumn{2}{c}{ Section~\ref{subsection:hgj-quantum-second} } & \multicolumn{2}{c}{With more RAM} \\
bound & Time & Memory & Time & Memory & Time & Memory \\
\midrule
0.0500 & 0.4433 & 0.0501 & 0.4433 & 0.0501 & 0.4412 & 0.0650\\
0.1000 & 0.3896 & 0.1000 & 0.3896 & 0.1000 & 0.3860 & 0.1259\\
0.1500 & 0.3348 & 0.1501 & 0.3348 & 0.1501 & 0.3301 & 0.1894\\
0.3000 & 0.2374 & 0.2373 & 0.2356 & 0.2356 & 0.2373 & 0.2373\\
\bottomrule
\end{tabular}
\end{table}

\section{New Algorithms Based on Quantum Walks}
\label{section:new-qw}


In this section, we improve the algorithm by Helm and May~\cite{DBLP:conf/tqc/HelmM18} based on BCJ and the MNRS quantum walk framework. Our algorithm is a quantum walk on a product Johnson graph, as in Section~\ref{subsection:subsumqw}. There are two new ideas involved.


\subsection{Asymmetric 5th level}

In our new algorithm, we can afford one more level than BCJ. We then have a 6-level merging tree, with levels numbered 5 down to 0. Lists at level $i$ all have the same size $\ell_i$, \emph{except at level 5}. Recall that the merging tree, and all its lists, is the additional data structure attached to a node in the Johnson graph. In the original algorithm of~\cite{DBLP:conf/tqc/HelmM18}, there are 5 levels, and a node is a collection of $16$ lists, each list being a subset of size $\ell_4$ among the $g(1/16+\alpha_3,\alpha_3)/2$ vectors having the right distribution.



In our new algorithm, at level 5, we separate the lists into ``left'' lists of size $\ell_5^l$ and ``right'' lists of size $\ell_5^r$. The quantum walk will only be performed on the left lists, while the right ones are full enumerations. Each list at level 4 is obtained by merging a ``left'' and a ``right'' list. The left-right-split at level 5 is then asymmetric: vectors in one of the left lists $L_5^l$ are sampled from $D^{\eta n}[\alpha_4, 1/32, \gamma_4] \times \{0^{(1 -\eta) n}\}$ and the right lists $L_5^r$ contain \emph{all} the vectors from $\{0^{\eta n}\} \times D^{(1-\eta) n}[\alpha_4, 1/32, \gamma_4]$. This yields a new constraint: $\ell_5^r = f(1/32+\alpha_4-2\gamma_4, \alpha_4, \gamma_4) (1 - \eta)$.



While this asymmetry does not bring any advantage classically, it helps in reducing the update time. We enforce the constraint $\ell_5^r = c_4$, so that for each element of $L_5^l$, there is on average one matching element in $L_5^r$. So updating the list $L_4$ at level 4 is done on average time 1. Then we also have $\ell_4 = \ell_5^l$.


With this construction, $\ell_5^r$ and $\ell_5^l$ are actually unneeded parameters. We only need the constraints $c_4 (= \ell_5^r) = f(1/32+\alpha_4-2\gamma_4, \alpha_4, \gamma_4) (1 - \eta)$ and $\ell_4 (= \ell_5^l) \leq f(1/32+\alpha_4-2\gamma, \alpha_4, \gamma_4) \eta$. The total setup time is now:
\begin{multline*}
\mathsf{S} = \max \bigg( \undertext[1]{Lv. 5 and 4}{ c_4, \ell_4 },~\undertext{Level 3}{2\ell_4  - (c_3 - c_4)},~\undertext{Level 2}{2\ell_3  - (c_2 - c_3)},~\undertext{Level 1}{2\ell_2 - (c_1 - c_2)}, \\
\undertext{Level 0}{\ell_1 + \max(\ell_1 - (1-c_1), 0) } \bigg)
\end{multline*}
and the expected update time for level 5 (inserting a new element in a list $L_5^l$ at the bottom of the tree) and at level 4 (inserting a new element in $L_4$) is 1. The spectral gap of the graph is $\delta = - \ell_5^l $
and the proportion of marked vertices is $\epsilon = -\ell_0$.


\paragraph{Saturation Constraints.}
In the quantum walk, we have $\ell_0 < 0$, since we expect only some proportion of the nodes to be marked (to contain a solution). This proportion is hence $\ell_0$. The saturation constraints are modified as follows:
\begin{align*}
\begin{array}{ll}
\ell_5^l \leq \frac{\ell_0}{16} + f(\frac{1}{32} + \alpha_4 - 2\gamma_4, \alpha_4, \gamma_4) \eta, & \ell_4 \leq \frac{\ell_0}{16} + f(\frac{1}{32} + \alpha_4 - 2\gamma_4, \alpha_4, \gamma_4) - c_4\\
\ell_3 \leq \frac{\ell_0}{8} + f(\frac{1}{16} + \alpha_3 - 2\gamma_3, \alpha_3, \gamma_3) - c_3,  & \ell_2 \leq \frac{\ell_0}{4} + f(\frac{1}{8} + \alpha_2 - 2\gamma_2, \alpha_2, \gamma_2) - c_2 \\
\ell_1 \leq \frac{\ell_0}{2} + f(1/4 + \alpha_1 - 2\gamma_1, \alpha_1, \gamma_1) - c_1 & \\
\end{array}
\end{align*}

Indeed, the \emph{classical} walk will go through a total of $-\ell_0$ trees before finding a solution. Hence, it needs to go through $-\ell_0/16$ different lists at level 5 (and 4), which is why we need to introduce $\ell_0$ in the saturation constraint: there must be enough elements, not only in $L_5^l$, but in the whole search space that will be spanned by the walk. These constraints ensure the existence of marked vertices in the walk.

%
%

\subsection{Better Setup and Updates using quantum search}\label{subsection:update-qsearch}

Along the lines of Lemma~\ref{lemma:quantum-matching-filtering} and corollary~\ref{lemma:symmetric-merging}, we now show how to use a quantum search to speed up the \textsf{Setup} and \textsf{Update} steps in the quantum walk. As the structure of the graph is unchanged, we still have $\epsilon = - \ell_0$ and a spectral gap $\delta = - \ell_5^l$.

\paragraph{Setup.} 
Let $p_i, (1 \leq i \leq 3)$ be the filtering probabilities at level $i$, \emph{i.e.} the (logarithms of the) probabilities that an element that satisfies the modulo condition resp. at level $i$ also has the desired distribution of $0$s, $1$s, $-1$s and $2$s, and appears in list $L_i$. Notice that $p_i \leq 0$. Due to the left-right split, there is no filtering at level 4.

We use quantum filtering (Corollary~\ref{lemma:symmetric-merging}) to speed up the computation of lists at levels 3, 2 and 1 in the setup, reducing in general a time $2 \ell - c$ to $2 \ell -c + \pf/2$.
It does not apply for level 0, since $L_0$ has a negative expected size. At this level, we will simply perform a quantum search over $L_1$. If there is too much constraint, \emph{i.e.} $(1-c_1) > \ell_1$, then for a given element in $L_1$, there is on average less than one modular candidate. If $(1-c_1) < \ell_1$, there is on average more than one (although less than one with the filter) and we have to do another quantum search on them all. This is why the setup time at level 0, in full generality, becomes $(\ell_1 + \max(\ell_1 - (1-c_1), 0))/2$. The setup time can thus be improved to:
\begin{multline*}
\mathsf{S} = \max \bigg( \undertext[1]{Lv. 5 and 4}{c_4, \ell_4}, \undertext{Level 3}{2\ell_4 - (c_3- c_4) + p_3/2}, \undertext{Level 2}{2\ell_3 - (c_2- c_3) + p_2/2}, \\
\undertext{Level 1}{2\ell_2 - (c_1 - c_2) + p_1/2}, \undertext{Level 0}{ (\ell_1 + \max(\ell_1 - (1 - c_1), 0))/2 } \bigg) \enspace.
\end{multline*}


\paragraph{Update.} Our update will also use a quantum search. First of all, recall that the updates of levels 5 and 4 are performed in (expected) time 1. Having added an element in $L_4$, we need to update the upper level. There are on average $\ell_4 - (c_3 - c_4)$ candidates satisfying the modular condition. To avoid a blowup in the time complexity, we forbid to have more than one element inserted in $L_3$ on average, which means: $\ell_4 - (c_3 - c_4) + p_3 \leq 0 \iff \ell_3 \leq \ell_4$. We then find this element, if it exists, with a quantum search among the $\ell_4 - (c_3 - c_4)$ candidates. 

Similarly, as at most one element is updated in $L_3$, we can move on to the upper levels 2, 1 and 0 and use the same argument. We forbid to have more than one element inserted in $L_2$ on average: $\ell_3 - (c_2 - c_3) + p_2 \leq 0 \iff \ell_2 \leq \ell_3$, and in $L_1$: $\ell_1 \leq \ell_2$. At level 0, a quantum search may not be needed, hence a time $\max(\ell_1 - (1 - c_1), 0) / 2$\XB{Pas clair.}. The expected update time becomes:
\begin{multline*}
 \mathsf{U} =  \max \bigg( 0, \undertext{Level 3}{(\ell_4 - (c_3 - c_4))/2}, \undertext{Level 2}{(\ell_3 - (c_2 - c_3))/2},\\
\undertext{Level 1}{(\ell_2 - (c_1 - c_2))/2},~\undertext{Level 0}{(\ell_1 - (1-c_1))/2} \bigg) \enspace.
\end{multline*}


\subsection{Parameters}\label{subsection:new-quantum-bcj}
Using the following parameters, we found an algorithm that runs in time $\bigOt{2^{0.2156n}}$:
\begin{gather*}
\ell_0 = -0.1916, \ell_1 = 0.1996, \ell_2 = 0.2030, \ell_3 = 0.2110, \ell_4 (=\ell_5^l) = 0.2110 \\
c_1 = 0.6190, c2 = 0.4445, c3 =0.2506, c_4 (=\ell_5^r) = 0.0487 \\
\alpha_1 = 0.0176, \alpha_2 = 0.0153, \alpha_3 = 0.0131, \alpha_4 = 0.0087 \\
\gamma_1 = 0.0019, \gamma_2 = \gamma_3 = \gamma_4 = 0, \eta = 0.8448
\end{gather*}


There are many different parameters that achieve the same time. The above set achieves the lowest memory that we found, at $\bigOt{2^{0.2110n}}$. Note that time and memory complexities are different in this quantum walk, contrary to previous works, since the update procedure has now a (small) exponential cost.

\begin{remark}[Time-memory tradeoffs]
Quantum walks have a natural time-memory tradeoff which consists in reducing the vertex size. Smaller vertices have a smaller chance of being marked, and the walk goes on for a longer time. This is also applicable to our algorithms, but requires a re-optimization with a memory constraint.
\end{remark}

\section{Mitigating Quantum Walk Heuristics for Subset-Sum}
\label{section:fix-qw}

In this section, we provide a modified quantum walk \textsc{NEW-QW} for any quantum walk subset-sum algorithm \textsc{QW}, including \cite{DBLP:conf/pqcrypto/BernsteinJLM13,DBLP:conf/tqc/HelmM18} and ours, that will no longer rely on Heuristic~\ref{heuristic:helm-may}. In \textsc{NEW-QW}, the Johnson graph is the same, but the vertex data structure and the update procedure are different (Section~\ref{subsection:vertex-data-struct}). It allows us to guarantee the update time, at the expense of losing some marked vertices. In Section~\ref{subsection:fraction-marked}, we will show that most marked vertices in \textsc{QW} remain marked.


\subsection{New Data Structure for Storing Lists}\label{subsection:new-data-struct}

The main requirement of the vertex data structure is to store lists of subknapsacks with modular constraints in QRAQM. For each list, we will use two data structures. The first one is the combination of a hash table and a skip list given in~\cite{DBLP:journals/siamcomp/Ambainis07} (abbreviated \emph{skip list} below) and the second one is a \emph{Bucket-modulus list} data structure, adapted from Definition~\ref{def:unique-mod}, that we define below.

\paragraph{Hash Table and Skip List.}
We use the data structure of~\cite{DBLP:journals/siamcomp/Ambainis07} to store lists of entries $(\vec{e}, \vec{e} \cdot \vec{a})$, sorted by knapsack value $\vec{e} \cdot \vec{a}$. The data structure for $M$ entries, that we denote $\mathcal{SL}(M)$, uses $\bigOt{M}$ qRAM memory cells and supports the following operations: inserting an entry in the list, deleting an entry from the list and producing the uniform superposition of entries in the list. All these operations require time $\mathsf{polylog}(M)$.

We resort to this data structure because the proposal of ``radix trees'' in~\cite{DBLP:conf/pqcrypto/BernsteinJLM13} is less detailed. It is defined relatively to a choice of $\mathsf{polylog}(M) = \poly(n)$ hash functions selected from a family of independent hash functions of the entries (we refer to~\cite{DBLP:journals/siamcomp/Ambainis07} for more details). For a given choice of hash functions, the insertion or deletion operations can fail. Thus, the data structure is equipped with a superposition of such choices. Instead of storing $\mathcal{SL}(M)$, we store: $\sum_h \ket{h} \ket{\mathcal{SL}_h(M)}$ where $\mathcal{SL}_h$ is the data structure flavored with the choice of hash functions $h$. Insertions and deletions are performed depending on $h$. This allows for a globally negligible error: if sufficiently many hash functions are used, the insertion and deletion of \emph{any element} add a global error vector of amplitude $o(2^{-n})$ \emph{regardless of the current state of the data}. The standard ``hybrid argument'' from~\cite{DBLP:journals/siamcomp/BennettBBV97} and~\cite[Lemma 5]{DBLP:journals/siamcomp/Ambainis07} can then be used in the context of an MNRS quantum walk. 

\begin{proposition}[\cite{DBLP:journals/siamcomp/Ambainis07}, Lemma 5, adapted]
Consider an MNRS quantum walk with a ``perfect'' (theoretical) update unitary $U$, managing data structures, and an ``imperfect'' update unitary $U'$ such that, for any basis state $\ket{x}$:
$$ U' \ket{x} = U \ket{x} + \ket{\delta_x}  $$
where $\ket{\delta_x}$ is an error vector of amplitude bounded by $o(2^{-n})$ \emph{for any $x$}. Then running the walk with $U'$ instead of $U$, after $T$ steps, the final ``imperfect'' state $\ket{\psi'}$ deviates from the ``perfect'' state $\ket{\psi}$ by: $\| \ket{\psi'} - \ket{\psi} \| \leq o(2^{-n} T) $.
\end{proposition}

This holds as a general principle: in the update unitary, any perfect procedure can be replaced by an imperfect one as long as its error is negligible (with respect to the total number of updates) \emph{and data-independent}. In contrast, the problem with Heuristic~\ref{heuristic:helm-may} is that a generic constant-time update induces data-dependent errors (bad cases) that do not seem easy to overcome.


\paragraph{Bucket-modulus List.}
Let $B = \poly(n)$ be a ``bucket size'' that will be chosen later. The bucket-modulus list is a tool for making our update time data-independent: it limits the number of vectors that can have a given modulus (where moduli are of the same order as the list size).

\begin{definition}[Bucket-modulus list]\label{def:bucket-mod}
A Bucket-modulus list $\mathcal{BL}(B, M)$ is a qRAM data structure that stores at most $B \times M$ entries $(\vec{e}, \vec{e} \cdot \vec{a})$, with at most $B$ entries sharing the same modulus $\vec{e} \cdot \vec{a} \mod M$. Thus, $\mathcal{BL}(B, M)$ contains $M$ ``buckets''. Buckets are indexed by moduli, and kept sorted. It supports the following operations:
\begin{itemize}
\item Insertion: insert $(\vec{e}, \vec{e} \cdot \vec{a})$. If the bucket at index $\vec{e} \cdot \vec{a} \mod M$ contains $B$ elements, \emph{empty the bucket}. Otherwise, sort it using a simple sorting circuit.
\item Deletion: remove an entry from the corresponding bucket.
\item Query in superposition: similar as in Definition~\ref{def:unique-mod}.
\end{itemize}
\end{definition}

In our new quantum walks, each list will be stored in a skip list $\mathcal{SL}(M)$ associated with a bucket-modulus $\mathcal{BL}(B, M)$. Each time we insert or delete an element from $\mathcal{SL}(M)$, we update the bucket-modulus list accordingly, according to the following rules.

Upon deletion of an element $\vec{e}$ in $\mathcal{SL}(M)$, let $\vec{e} \cdot \vec{a} = T \mod{M}$, there are three cases for $\mathcal{BL}(B, M)$:
\begin{itemize}
\item If $|\{ \vec{e'} \in \mathcal{SL}(M), \vec{e'} \cdot \vec{a} = T \}| > B+1$, then bucket number $T$ in $\mathcal{BL}(B,M)$ stays empty;
\item If $|\{ \vec{e'} \in \mathcal{SL}(M), \vec{e'} \cdot \vec{a} = T \}| = B + 1$, then removing $\vec{e}$ makes the number of elements reach the bound $B$, so we add them all in the bucket at index $T$;
\item If $|\{ \vec{e'} \in \mathcal{SL}(M), \vec{e'} \cdot \vec{a} = T \}| \leq B$, then we remove $\vec{e}$ from its bucket.
\end{itemize}

Upon insertion of an element $\vec{e}$ in $\mathcal{SL}(M)$, there are also three cases for $\mathcal{BL}(B, M)$:
\begin{itemize}
\item If $|\{ \vec{e'} \in \mathcal{SL}(M), \vec{e'} \cdot \vec{a} = T \}| = B$, then we empty the bucket at index $T$;
\item If $|\{ \vec{e'} \in \mathcal{SL}(M), \vec{e'} \cdot \vec{a} = T \}| < B$, then we add $\vec{e}$ to the bucket at index $T$ in $\mathcal{BL}(B,M)$;
\item If $|\{ \vec{e'} \in \mathcal{SL}(M), \vec{e'} \cdot \vec{a} = T \}| > B$, then the bucket is empty and remains empty.
\end{itemize}

In all cases, there are at most $B$ insertions or deletions in a single bucket. Note that $\mathcal{BL}(B,M) \subseteq \mathcal{SL}(M)$ but that some elements of $\mathcal{SL}(M)$ will be dropped. 

\begin{remark}
The mapping from a skip list of size $M$ (considered as perfect), which does not ``forget'' any of its elements, to a corresponding bucket-modulus list with $M$ buckets, which forgets some of the previous elements, is deterministic. Given a skip list $L$, a corresponding bucket modulus list $L'$ can be obtained by inserting all elements of $L$ into an empty bucket modulus list.
\end{remark}

\subsection{New Data Structure for Vertices}
\label{subsection:vertex-data-struct}

The algorithms that we consider use multiple levels of merging. However, we will focus only on a single level. Our arguments can be generalized to any constant number of merges (with an increase in the polynomial factors involved). Recall that the product Johnson graph on which we run the quantum walk is unchanged, only the data structure is adapted.



In the following, we will consider the merging of two lists $L_l$ and $L_r$ of subknapsacks of respective sizes $\ell_l$ and $\ell_r$, with a modular constraint $c$ and a filtering probability $\pf$. The merged list is denoted $L^c= L_l \bowtie_c L_r$ and the filtered list is denoted $L^f$. We assume that pairs $(\vec{e_1}, \vec{e_2})$ in $L^c$ must satisfy $(\vec{e_1} + \vec{e_2}) \cdot \vec{a} = 0 \mod{2^{cn}}$ (the generalization to any value modulo any moduli is straightforward).

On the positive side, our new data structure can be updated, \emph{by design}, with a fixed time that is data-independent. On the negative side, we will not build the complete list $L^f$, and miss some of the solutions. As we drop a fraction of the vectors, some nodes that were previously marked will potentially appear unmarked, but this fraction is polynomial at most. We defer a formal proof of this fact to Section~\ref{subsection:fraction-marked} and focus on the runtime.

We will focus on the case where $\ell_l = \ell_r$ and either $L_l$ or $L_r$ are updated, which happens at all levels in our quantum walk, except the first level. Because there is no filtering at the first level, it is actually much simpler to study with the same arguments. In previous quantum walks, we had $\ell^c =  2\ell - c \leq \ell$, \emph{i.e.} $\ell \leq c$; now we will have $2\ell - c \geq \ell$ and $2 \ell - c + \pf \leq \ell$.

%
%

Recall that our heuristic time complexity analysis assumes an update time $(\ell-c)/2$. Indeed, the update of an element in $L_l$ or $L_r$ modifies on average $(\ell-c)$ elements in $L_l \bowtie_c L_r$, among which we expect at most one filtered pair $(\vec{e_1}, \vec{e_2})$ (by the inequality $2 \ell - c + \pf \leq \ell$). We find this solution with a quantum search. In the following, we modify the data structure of vertices in order to guarantee the best update time possible, up to additional polynomial factors. We will see however that it does not reach $(\ell-c)/2$. We now define our intermediate lists and sublists, before giving the update procedure and its time complexity.


\paragraph{Definitions.}
Both lists $L_l, L_r$ are of size $M \simeq 2^{\ell n}$. We store them in skip lists. In both $L_r$ and $L_l$, for each $T \leq M$, we expect on average only one element $\vec{e}$ such that $\vec{e} \cdot \vec{a}  = T \mod{M}$. We introduce two \emph{Bucket-modulus lists} (Definition~\ref{def:bucket-mod}) $L_l'(B,M)$ and $L_r'(B,M)$ that we will write as $L_l'$ and $L_r'$ for simplicity, indexed by $\vec{e} \cdot \vec{a} \mod{M}$, with an arbitrary bound $B = \poly(n)$ for the bucket sizes. They are attached to $L_l$ and $L_r$ as detailed in Section~\ref{subsection:new-data-struct}. When an element in $L_l$ or $L_r$ is modified, they are modified accordingly.

In $L_l'$ and $L_r'$, we consider the sublists of subknapsacks having the same modulo $C \mod{2^{cn}}$, and we denote by $L_{l,C}'$ and $L_{r,C}'$ these sublists. They can be easily considered separately since the vectors are sorted by knapsack weight. By design of the bucket-modulus lists, $L_{l,C}'$ and $L_{r,C}'$ both have size at most $B 2^{(\ell-c)n}$. We have:
$$ L_l' \bowtie_c L_r' = \bigcup_{0 \leq C \leq 2^{cn} - 1} L_{l,C}' \times L_{r,C}' \enspace.$$

Next, we have a case disjunction to make. The most complicated case is when $2\ell - 2c + \pf > 0$, that is, each product $L_{l,C}' \times L_{r,C}'$ for a given $C$ yields more than one filtered pair on average. In that case, we define sublists $L_{l,C,i}'$ of $L_{l,C}'$ and sublists $L_{r,C,j}'$ of $L_{r,C}'$ using a new arbitrary modular constraint, so that each of these sublists is of size $-\pf/2$ (at most). There are $\ell-c + \pf/2$ sublists (exactly). The rationale of this cut is that a product $L_{l,C,i}' \times L_{r,C,j}'$ for a given $i,j$ now yields on average a single filtered pair (or less). When $2\ell - 2c + \pf \leq 0$, we don't perform this last cut and consider the product $L_{l,C}' \times L_{r,C}'$ immediately. By a slight abuse of notation, we denote: $(L_{l,C,i}' \times L_{r,C,j}')^f$ the set of filtered pairs from $L_{l,C,i}' \times L_{r,C,j}'$, and we have:
$$ L^f =  \bigcup_{0 \leq C \leq 2^{cn} - 1} \bigcup_{i,j}(L_{l,C,i}' \times L_{r,C,j}')^f \enspace. $$


\begin{algorithm}[htb]
\caption{Update algorithm: given $L_l, L_r$ of size $\ell$, we insert or delete an element in $L_l$ and update the filtered list $L^f$ accordingly. We focus here on the case $2\ell - 2c + \pf > 0$.}\label{algo:update}
\begin{algorithmic}[1]
\Statex \textbf{Data: } skip lists for $L_l, L_r, L^f$, bucket-modulus lists $L_l', L_r'$ 
\State \Comment{The bucket-modulus list for $L^f$ will be updated later}
\Statex \textbf{Input: } an insertion / deletion instruction for $L_l$
\Statex \textbf{Output: } updates $L_l, L_l', L^f$ accordingly
\State Insert or delete in $L_l$ \Comment{only one element to update}
\State Update the bucket-modulus structure $L_l'$ \Comment{at most $B$ elements to update}

\For{each element $\vec{e}$ to insert / delete in $L_l'$} \Comment{$B = \poly(n)$ iterations}
\State Select its corresponding sublist $L_{l,C,i}'$
\State Let $L_{l,C,i}'' = L_{l,C,i}' \cup \{ \vec{e} \}$ or $L_{l,C,i}' \backslash \{\vec{e}\}$
\For{each sublist $L_{r,C,j}'$}\Comment{$\ell-c + \pf/2$ iterations}
\State Estimate $s = |(L_{l,C,i}' \times L_{r,C,j}')^f|$ \Comment{time $\bigOt{B \times 2^{- \pf n/2}}$}
\State Estimate $s' = |(L_{l,C,i}'' \times L_{r,C,j}')^f|$ \Comment{time $\bigOt{B \times 2^{- \pf n/2}}$}
\Statex \Comment{In the case of an insertion, $s' \geq s$ and $s' \leq s$ for a deletion}
\State \textbf{if} $s > B$ and $s' \leq B$
\Statex \Comment{The removal of $\vec{e}$ makes the number of filtered pairs acceptable}
\State ~~ ~~ \textbf{then} $L^f \leftarrow L^f \cup (L_{l,C,i}'' \times L_{r,C,j}')^f$
\State \textbf{if} $s > B$ and $s' > B$
\State ~~ ~~ \textbf{then} do nothing
\State \textbf{if} $s \leq B$ and $s' > B$ 
\Statex \Comment{The insertion of $\vec{e}$ overflows the filtered pairs}
\State ~~ ~~ \textbf{then} remove all $(L_{l,C,i}' \times L_{r,C,j}')^f$ from $L^f$
\State \textbf{if} $s \leq B$ and $s' \leq B$
\State ~~ ~~ \textbf{then} update $L^f$ with the (at most) $B$ new or removed pairs
\EndFor
\EndFor

\end{algorithmic}
\end{algorithm}



\paragraph{Algorithm and Complexity.}
Algorithm~\ref{algo:update} details our update procedure. We now compute its time complexity and explain why it remains data-independent. Recall that we want to avoid the ``bad cases'' where an update goes on for too long: this is the case where an update in $L_l$ (or $L_r$) creates too many updates in $L^f$. In Algorithm~\ref{algo:update}, we avoid this by deliberately limiting the number of elements that can be updated. We can see that $L^f$ will be smaller than the ``perfect'' one for two reasons: $\bullet$~the bucket-modulus data structure loses some vectors, since the buckets are dropped when they overflow. $\bullet$~filtered pairs are lost. Indeed, the algorithm ensures that in $L^f$, at most $B$ solutions $\vec{e}_l + \vec{e}_r$ come from a cross-product $L_{l,C,i}' \times L_{r,C,j}'$.

This makes the update procedure \emph{history-independent} and its time complexity \emph{data-independent}. Indeed:

\begin{lemma}
The state of the data structures $L_l, L_r, L^f$ after Algorithm~\ref{algo:update} depends only on $L_l, L_r, L^f$ before and on the element that was inserted / deleted.
\end{lemma}

We omit a formal proof, as it follows from our definition of the bucket-modulus list and of Algorithm~\ref{algo:update}.

\begin{lemma}
With a good choice of $B$, Algorithm~\ref{algo:update} runs with a data-independent error in $o(2^n)$. The time complexity of Algorithm~\ref{algo:update} is $\bigOt{2^{(\ell-c)n}}$ and an update modifies $\bigOt{2^{\max(\ell-c + \pf/2, 0)n}}$ elements in the filtered list $L^f$ at the next level (respectively $\ell-c$ and $\max(\ell-c + \pf/2, 0)$ in log scale).
\end{lemma}

\begin{proof}
We check step by step the time complexity of Algorithm~\ref{algo:update}:
\begin{itemize}
\item Insertion and deletion from the skip list for $L_l$ is done in $\poly(n)$, with a global error that can be omitted.
\item The bucket-modulus list $L_l'$ is updated in time $\bigO{B} = \poly(n)$ without errors. At most $B$ elements must be inserted or removed.
\item For each insertion or removal in $L_l'$, we select the corresponding sublist $L_{l,C,i}'$ (or simply $L_{l,C}'$ if $2\ell - 2c + \pf \leq 0$). We look at the sublists $L_{r,C,j}'$ and we estimate the number of filtered pairs in the products $L_{l,C,i}' \times L_{r,C,j}'$ (of size $-\pf$), checking whether it is smaller or bigger than $B$. We explain in \ifeprint Appendix~C\else \cite[Appendix~C]{DBLP:journals/iacr/BonnetainBSS20} \fi how to do that reversibly in time $\bigOt{B \times 2^{-\pf n/2}}$ ($-\pf/2$ in log scale).
There are $\ell-c + \pf/2$ classical iterations, thus the total time is $\ell-c$.
\item Depending whether we have found more or less than $B$ filtered pairs, we will have to remove or to add all of them in $L^f$. This means that $B \times 2^{(\ell-c + \pf/2)n}$ insertion or deletion instructions will be passed over to $L^f$.
\end{itemize}

There are two sources of data-independent errors: first, the skip list data structure (see Section~\ref{subsection:new-data-struct}). Second, the procedure of \ifeprint Appendix~C\else \cite[Appendix~C]{DBLP:journals/iacr/BonnetainBSS20}\fi. Both can be made exponentially small at the price of a polynomial overhead. Note that $B$ will be set in order to get a sufficiently small probability of error (see the next section), and can be a global $\bigO{n}$. However, the polynomial overhead of our update unitary grows with the number of levels. \qed
\end{proof}

\subsection{Fraction of Marked Vertices}
\label{subsection:fraction-marked}

Now that we have computed the update time of \textsc{NEW-QW}, it remains to compute its fraction $\epsilon_{new}$ of marked vertices. We will show that $\epsilon_{new}= \epsilon \left(1 - \frac{1}{\poly(n)} \right)$ with overwhelming probability on the random subset-sum instance, where $\epsilon$ is the previous fraction in \textsc{QW}.

Consider a marked vertex in \textsc{QW}. There is a path in the data structure leading to the solution, hence a constant number of subknapsacks $\vec{e_1}, \ldots, \vec{e_t}$ such that the vertex will remain marked \emph{if and only} if none of them is ``accidentally'' discarded by our new data structure. Thus, if $G$ is the graph of the walk, we want to upper bound:
$$ \Pr_{v \in G} \left( \begin{matrix}
\text{$v$ is marked in \textsc{QW} and} \\ \text{not marked in \textsc{NEW-QW}}
\end{matrix} \right) \leq \sum_{ \vec{e_i}, 1 \leq i \leq t} \Pr_{v \in G} \left(  \begin{matrix}
\vec{e_i} \in v \text{ in \textsc{QW}} \\
\vec{e_i} \notin v \text{ in \textsc{NEW-QW}} \\
\end{matrix} \right) \enspace. $$

We focus on some level in the tree, on a list $L$ of average size $2^{\ell n}$, and on a single vector $\vec{e_0}$ that must appear in $L$. Subknapsacks in $L$ are taken from $\mathcal{B}\subseteq D^n[\alpha,\beta, \gamma]$. We study the event that $\vec{e_0}$ is accidentally discarded from $L$. This can happen for two reasons:
\begin{itemize}
\item we have $|\{\vec{e} \in L,  \vec{e}\dotprod \vec{a} = \vec{e_0}\dotprod \vec{a} \mod{2^{\ell n}}\}| > B$: the vector is dropped at the bucket-modulus level;
\item at the next level, there are more than $B$ pairs from some product of lists $L_{l,C,i}' \times L_{r,C,j}'$ to which the vector $\vec{e_0}$ belongs, that will pass the filter.
\end{itemize}

We remark the following to make our computations easier.

\begin{fact}
We can replace the $L$ from our new data structure \textsc{NEW-QW} by a list of exact size $2^{\ell n}$, which is a sublist from the list $L$ in \textsc{QW}.
\end{fact}

At successive levels, our new data structure discards more and more vectors. Hence, the actual lists are smaller than in \textsc{QW}. However, removing a vector $\vec{e}$ from a list, if it does not unmark the vertex, does not increase the probability of unmarking it at the next level, since $\vec{e}$ does not belong to the unique solution.

\begin{fact}
When a vertex in \textsc{NEW-QW} is sampled uniformly at random, given a list $L$ at some merging level, we can assume that the elements of $L$ are sampled uniformly at random from their distribution $\mathcal{B}$ (with a modular constraint).
\end{fact}

This fact translates Heuristic~\ref{heur1} as a global property of the Johnson graph. At the first level, nodes contain lists of exponential size which are sampled without replacement. However, when sampling with replacement, the probability of collisions is exponentially low. Thus, we can replace $\Pr_{v \in G}$ by $\Pr_{v \in G'}$ where $G'$ is a ``completed'' graph containing all lists sampled uniformly at random with replacement. This adds only a negligible number of vertices and does not impact the probability of being discarded.


\paragraph{Number of Vectors Having the Same Modulus.}
Let $N \simeq 2^n$ and $M$ be a divisor of $N$. 
Given a particular $\vec{e}_0\in\mathcal{B}$ and a vector $\vec{a}\in\mathbb{Z}_{N}^n$, 
$$ \text{For $\vec{e}\in \mathcal{B}$, define } X_\vec{e}(a)= \begin{cases}
  1 & \text{ if } \vec{e}\dotprod \vec{a} = \vec{e}_0 \dotprod \vec{a} \pmod M\\
  0 & \text{ otherwise}
\end{cases}$$

We prove the following Lemma in \ifeprint Appendix~\ref{appendix:probas} \else the full version of the paper~\cite{DBLP:journals/iacr/BonnetainBSS20}\fi.

\begin{lemma}\label{lemma:prob}
If $|\mathcal{B}|\gg M \simeq |L|$, then for a $1-\negl(n)$ proportion of $\vec{a} \in \mathbb{Z}_{N}^n$, and with an appropriate $B = \bigO{n}$:
\begin{equation}
\label{eq:prob}
\prob[\vec{e}_1, \cdots, \vec{e}_{|L|} \sim Unif(\mathcal{B})]{\sum_{i=1}^{|L|} X_{\vec{e}_i} (\vec{a})< B -1} > 1 - \frac{1}{\poly(n)}
\end{equation}
\end{lemma}


%


For the number of filtered pairs, we use the fact that the vectors at each level are sampled uniformly at random from their distribution. If this is the case, then a Chernoff bound (similar to the proof of Lemma~\ref{lemma:prob}) limits the deviation of the number of filtered pairs in $L_{l,C,i}' \times L_{r,C,j}'$ from its expectation (which is 1 by construction): the probability that there are more than $B + 1$ pairs is smaller than $e^{- (B+1)/3}$. By taking a sufficiently big $B = \bigO{n}$, we can take a union bound over all products of lists $L_{l,C,i}' \times L_{r,C,j}'$ in which $\vec{e_0}$ intervenes. We also take a union bound over the intermediate subknapsacks that we are considering. The loss of vertices remains inverse polynomial.

\subsection{Time Complexities without Heuristic~2}
\label{subsection:time-without-heuristics}

\newcommand{\maxz}{\widehat{ \mathrm{max}}}

Previous quantum subset-sum algorithms~\cite{DBLP:conf/pqcrypto/BernsteinJLM13,DBLP:conf/tqc/HelmM18} have the same time complexities without Heuristic~\ref{heuristic:helm-may}, as they fall in parameter ranges where the bucket-modulus data structure is enough. However, this is not the case of our new quantum walk. We keep the same set of constraints and optimize with a new update time. Although using the extended $\zomd$ representations brings an improvement, neither do the fifth level, nor the left-right split. This simplifies our constraints. Let $\maxz(\cdot) = \max(\cdot,0)$. The guaranteed update time becomes:
\begin{multline*}
\mathsf{U} = \maxz \bigg( \undertext{Level 2}{\ell_3 - (c_2 - c_3)}, \undertext{Number of elements to update at level 1}{\maxz(\ell_3 - (c_2 - c_3) + \frac{p_2}{2})} + \maxz(\ell_2 - (c_1-c_2)) ,\\
\underbrace{ \frac{1}{2} \left( \maxz \!\left( \ell_3\! - \!(c_2 \!-\! c_3) \!+ \!\frac{p_2}{2}\right) \! + \maxz \!\left(\ell_2 \!- \!(c_1 - c_2) + \frac{p_1}{2}\right)\! + \maxz(\ell_1 - (1-c_1)) \right) }_{\mbox{Final quantum search among all updated elements}} \bigg)
\end{multline*}

We obtain the time exponent $0.2182$ (rounded upwards) with the following parameters (rounded). The memory exponent is $0.2182$ as well.
\begin{gather*}
\ell_0 = -0.2021, \ell_1 = 0.1883, \ell_2 = 0.2102, \ell_3 = 0.2182, \ell_4 = 0.2182  \\
c_3 = 0.2182, c_2 = 0.4283, c_1 = 0.6305, p_0 = -0.2093, p_1 = -0.0298, p_2 = -0.0160 \\
\alpha_1 = 0.0172, \alpha_2 = 0.0145, \alpha_3 = 0.0107, \gamma_1 = 0.0020
\end{gather*}


\section{Conclusion}
\label{section:conclusion}

In this paper, we proposed improved classical and quantum heuristic algorithms for subset-sum, building upon several new ideas. First, we used extended representations ($\zomd$) to improve the current best classical and quantum algorithms. In the quantum setting, we showed how to use a quantum search to speed up the process of \emph{filtering} representations, leading to an overall improvement on existing work. We built an ``asymmetric HGJ'' algorithm that uses a nested quantum search, leading to the first quantum speedup on subset-sum in the model of \emph{classical memory with quantum random access}. By combining all our ideas, we obtained the best quantum walk algorithm for subset-sum in the MNRS framework. Although its complexity still relies on Heuristic~\ref{heuristic:helm-may}, we showed how to partially overcome it and obtained the first quantum walk that requires only the classical subset-sum heuristic, and the best to date for this problem. 

\subsubsection*{Open Questions.}
We leave as open the possibility to use representations with ``-1''s (or even ``2''s) in a quantum asymmetric merging tree, as in Section~\ref{subsection:hgj-quantum-second}. Another question is how to bridge the gap between heuristic and non-heuristic quantum walk complexities. In our work, the use of an improved vertex data structure seems to encounter a limitation, and we may need a more generic result on quantum walks, similar to~\cite{DBLP:journals/mst/Ambainis10}. Finally, it would be of interest to study representations with a larger set of integers.


\XB{Vérifier si ça ne donne pas trop l'impression qu'on a eu la flemme de finir le boulot.}


\subsubsection*{Acknowledgments.}
The authors want to thank André Chailloux, Stacey Jeffery, Antoine Joux, Frédéric Magniez, Alexander May, Amaury Pouly, Nicolas Sendrier for helpful discussions and comments. Thanks to Zhenzhen Bao and the anonymous CRYPTO and ASIACRYPT referees for their detailed comments. This project has received funding from the European Research Council (ERC) under the European Union’s Horizon 2020 research and innovation programme (grant agreement no. 714294 - acronym QUASYModo). Research also supported in part by the ERA-NET Cofund in Quantum Technologies project QuantAlgo and the French ANR Blanc project RDAM.


\bibliography{biblio}
\bibliographystyle{splncs03}

\newpage
\section*{\huge{Appendices}}
\appendix

\section{Bijection between Representations and Integers}
\label{appendix:rpbijection}

It is well known
that there exists a bijection between $\left[0, {n \choose m} \right[$ and $n$-bit vectors of Hamming weight $m$, and this bijection can be computed in polynomial time in $n$~\cite{ch72}. In our case, $m = \beta n$ and such vectors are subknapsacks from $D^n[0, \beta]$. If $i_1, \ldots i_{m}$ are the bit-positions of the $m$ ``1'' in this vector, we map it to the $m$-tuple of integers: $(i_1, \ldots i_{m})$, and define the bijection as:
\begin{align*}
\phi~: (i_1, \ldots i_{m}) \mapsto {i_{m} - 1 \choose {m}} + \ldots + {i_1 - 1 \choose {1}}
\end{align*}
where ${i \choose {j}}$ has supposedly been precomputed for all $i \leq n, j \leq m$. In order to compute the inverse $\phi^{-1}$, we find for each $j \leq t$ the unique integer $i_j$ such that ${i_{j} - 1 \choose {m}} \leq x \leq {i_{j} \choose {m}}$.
We can generalize this to an arbitrary number of nonzero symbols (3 in our paper: ``1'', ``-1'' and ``2''), that we denote $1, 2, \ldots t$.  Let $k_1, \ldots k_t$ be the counts of each symbol in the vector $\vec{v}$. We map it to a tuple of tuples: $ (i^1_1, \ldots i^1_{k_1}), \ldots$, $(i^t_1, \ldots i^t_{k_t}) $
where the first vector represents the positions of the ``1'' among the $n$ bit positions, the second vector represents the positions of the ``2'' \emph{after having removed the ``1''}, and so on. Consequently, we have $0 \leq i^1_{j} \leq n-1$, $0 \leq i^2_{j} \leq n-1 - k_1$, \emph{etc}.

Next, we map each of these $t$ tuples individually to an integer, as was done above: $\phi^j (i^j_1, \ldots i^j_{k_j}) = x^j$ where $0 \leq x^j \leq {n - k_1 \ldots -k_{j-1} \choose k_j }$. Finally, we compute:
\begin{align*}
\phi(\vec{v}) &= x^1 + {n \choose k_1} x^2 + {n \choose k_1, k_2} x^3 + \ldots + {n \choose k_1, \ldots k_t} x^t \\
&= x^1 + {n \choose k_1} \left( x^2 + {n - k_1 \choose k_2} \left( x^3 + \ldots + {n - k_1 \ldots -k_{t-2} \choose k_{t-1} } x^t\right) \ldots \right) \enspace.
\end{align*}

$$  $$
By the bounds on the $x^j$, we remark that $0 \leq \phi(\vec{v}) \leq {n \choose k_1, \ldots k_t } - 1$. Furthermore, we observe that one can easily retrieve the $x^j$ by successive euclidean divisions, and use the bijections $\phi^j$ to finish the computation.


\begin{lemma}\label{lemma:representation-bijection}
Let $n, k_1, \ldots k_t$ be $t +1$ integers such that $k_1 + \ldots + k_t \leq n$. There exists a quantum unitary, realized with $poly(n)$ gates, that on input a number $j$ in $\left[0, {n \choose {k_1, \ldots k_t}} \right[$, writes on its output register the $j$-th vector $\phi^{-1}(j) \in \{0, \ldots, t\}^n$ having, for each $1 \leq i \leq t$, exactly $k_i$ occurrences of the symbol ``$i$''. There exists another unitary which writes, on input $\vec{v}$, the integer value $\phi(\vec{v})$.
\end{lemma}

Using this unitary in combination with a Quantum Fourier Transform, we can, for example, easily produce superpositions of subsets of $D^n[\alpha, \beta, \gamma]$, by taking arbitrary integer intervals.


\section{Filtering Probabilities}
\label{appendix:filter}

We give below the filtering probabilities for representations that use ``-1'' and ``2''. The principle is similar to Lemma~\ref{lemma:hgj-filter}, but the details are more technical. Notice that both of these lemmas assume symmetric input distributions, and thus are less general than Lemma~\ref{lemma:hgj-filter}.

\begin{lemma}[Filtering BCJ-style representations]
\label{lemma:bcj-filter}
Let $\vec{e_1}, \vec{e_2} \in D^n[\alpha, \beta]$ and $\gamma \leq 2 \alpha$. Then the logarithm of the probability that $\vec{e_1} + \vec{e_2} \in D^n[\gamma, 2\beta]$ is:
 \begin{multline*}
 \pfilterbcj{\alpha}{\beta}{\gamma} = \bin{\beta+\alpha}{\alpha-\gamma/2} + \bin{\alpha}{\alpha-\gamma/2} \\+ \trin{1-\beta-2\alpha}{\gamma/2}{\beta+\gamma/2} - \trin{1}{\beta+\alpha}{\alpha} 
 \end{multline*}
 \end{lemma}

\begin{proof}
In order to estimate the success probability, we need to estimate the number of well-formed representations, and how they can be decomposed. Given a fixed vector $\vec{e_1} \in D$, we count the number of compatible $\vec{e_2}$ such that $xn$ positions with a -1 from $\vec{e_2}$ are cancelled by a 1 from $\vec{e_1}$. As $\vec{e_1+e_2} \in D^n[\gamma,2\beta]$, there are
${(\beta+\alpha)n \choose xn} {\alpha n \choose (2\alpha - \gamma -x) n)} { (1-\beta-2\alpha)n \choose (\alpha-x)n, (\beta+\gamma-\alpha+x)n }$ such vectors.

Taking the logarithm and the standard approximations, its derivative is $\log\left(\frac{\beta + \alpha -x}x \frac{2\alpha-\gamma -x}{\gamma-\alpha+x} \frac{\alpha-x}{\beta+\gamma-\alpha+x}\right)$. This term is strictly decreasing for $O< x < \gamma-\alpha$, and equals 0 for $x = \alpha-\gamma/2$. Hence, this is the maximum, which correspond to the balanced case. It is equal, up to a polynomial loss, to the total number of compatible vectors. Hence, the log of the number of compatible vectors is

$\bin{\beta+\alpha}{\alpha-\gamma/2} + \bin{\alpha}{\alpha-\gamma/2} + \trin{1-\beta-2\alpha}{\gamma/2}{\beta+\gamma/2}.$
As there are $\trin{1}{\beta+\alpha}{\alpha}$ vectors in $D^n[\alpha,\beta]$, the lemma holds. \qed
\end{proof}

\begin{lemma}[Filtering representations using ``2''s]
\label{lemma:better-filter}
Let $\vec{e_1}, \vec{e_2} \in D^n[\alpha_1, \beta, \gamma_1]$ and $\alpha_0, \gamma_0 \geq 0$. Let us define :
$\left\lbrace\begin{array}{l}
x_{min} = \max(0, \alpha_1+\beta - \frac{1-\alpha_0+\gamma_0}{2}, \gamma_1 - \gamma_0/2)\\
x_{max} = \min(\alpha_1-\alpha_0/2, \alpha_0/2 + \beta_1 - \gamma_0, \gamma_1)
\end{array}\right.$. If $x_{min} \leqslant x_{max}$, then the logarithm of the probability that $\vec{e_1} + \vec{e_2} \in D^n[\alpha_0, 2\beta, \gamma_0]$ is at least:
 \begin{multline*}
 \pfilternous{\alpha_0}{\beta}{\gamma_0}{\alpha_1}{\gamma_1} = \max\limits_{x \in [x_{min}, x_{max}]} \trin{\alpha_1}{x}{\alpha_0/2} + \\ \trin{\alpha_1+\beta-2\gamma_1}{\gamma_0-2\gamma_1+2x}{\beta-\gamma_0-x+\alpha_0/2} + \bin{\gamma_1}{x} + \\ \quadrin{1-\beta-2\alpha_1+\gamma1}{\gamma_1-x}{\alpha_0/2}{\beta-\gamma_0-x+\alpha_0/2} -\\ \quadrin{1}{\alpha_1}{\alpha_1+\beta-2\gamma_1}{\gamma_1}
 \end{multline*}
\end{lemma}
\begin{proof}
To avoid the explosion of the number of variables, we restrain ourselves to the \textit{symmetric} cases (for which there are as much $1$s given by $0+1$ as $1$s given by $1+0$, etc.) Given some $\vec{e_1} \in D^n[\alpha_1, \beta, \gamma_1]$, we compute the number of compatible $\vec{e_2} \in D^n[\alpha_1, \beta, \gamma_1]$. These vectors can be sorted according to the number $x n$ of positions where a $-1$ from $\vec{e_1}$ cancels out a $2$ from $\vec{e_2}$. For $\vec{e_1} + \vec{e_2}$ to be in $D^n[\alpha_0, 2\beta, \gamma_0]$, we must have :
\begin{itemize}
\item $(-1) + (0) | (0) + (-1)$ : $\alpha_0 n/2$ times
\item $(-1) + (1) | (1) + (-1)$ : $(\alpha_1 - x - \alpha_0/2)n$ times
\item $(-1) + (2) | (2) + (-1)$ : $x n$ times
\item $(1) + (0) | (0) + (1)$ : $(\alpha_0/2 + \beta - \gamma_0 - x)n$ times
\item $(2) + (0) | (0) + (2)$ : $(\gamma_1 - x)n$ times
\item $(1) + (1)$ : $(\gamma_0 - 2\gamma_1 + 2 x)n$ times
\item $(0) + (0)$ : $(1 - \alpha_0 + \gamma_0 - 2\alpha_1 - 2\beta + 2x)n$ times (\emph{i.e.} the remaining)
\end{itemize}
Thus, in $\vec{e_1}$ (which contains $\alpha_1$ ``$-1$''s), $\alpha_0 n/2$ of the ``$-1$''s must match a ``$0$'', $(\alpha_1 - x - \alpha_0/2)n$ of the ``$-1$''s must match a ``$1$'' and the remaining $nx$ must match a ``$2$''. Therefore, there are ${\alpha_1 n \choose \alpha_0 n/2, x n}$ possible choices for the coordinates of $\vec{e_2}$ matching the ``$-1$''s of $\vec{e_1}$. Similarly, there are ${(\alpha_1+\beta-2\gamma_1) n \choose (\gamma_0-2\gamma_1+2x) n, (\beta-\gamma_0-x+\alpha_0/2) n}$ possibilities for the coordinates of $\vec{e_2}$ matching the ``$1$''s, ${\gamma_1 n \choose x n}$ for the coordinates matching the ``$2$''s, and ${(1-\beta-2\alpha_1+\gamma1) n \choose (\gamma_1-x) n, \alpha_0 n/2, (\beta-\gamma_0-x+\alpha_0/2)n}$ for the coordinates matching the ``$0$''s.

The total number of possibilities is :
\begin{multline*}
\sum\limits_x {\alpha_1 n \choose \alpha_0 n/2, x n} {(\alpha_1+\beta-2\gamma_1) n \choose (\gamma_0-2\gamma_1+2x) n, (\beta-\gamma_0-x+\alpha_0/2) n} \\ {\gamma_1 n \choose x n} {(1-\beta-2\alpha_1+\gamma1) n \choose (\gamma_1-x) n, \alpha_0 n/2, (\beta-\gamma_0-x+\alpha_0/2)n}
\end{multline*}

This quantity is defined only for $x_{min} \leqslant x \leqslant x_{max}$. If $x$ is outside of these bounds, one of these multinomial (at least) is zero and there is no compatible $\vec{e_2}$. As $x n$ must be an integer, there are only a linear number of possible choices for $x$. Therefore the number of all possible $\vec{e_2}$ is given, up to a polynomial factor, by the number of $\vec{e_2}$s for the best $x$.

In order to obtain a probability, we divide the number of compatible $\vec{e_2}$ by $n \choose \alpha_1 n, (\alpha_1 + \beta - 2\gamma_1)n, \gamma_1 n)$, which is the size of $D^n[\alpha_1, \beta, \gamma_1]$. We observe here that the logarithm of the probability that $\vec{e_2}$ is compatible with $\vec{e_1}$ is exactly $\pfilternous{\alpha_0}{\beta}{\gamma_0}{\alpha_1}{\gamma_1}$.

We only considered the symmetric cases, assuming that we can neglect the contribution of the asymmetric cases. If we cannot, it means that we underestimated the probability for $\vec{e_1}$ and $\vec{e_2}$ to be compatible, and we could improve further the parameters by taking into account the asymmetric cases as well. \qed
\end{proof}

\section{Estimating a Number of Solutions Reversibly}\label{appendix:nbsols}

We give a reversible procedure that, given a search space $X$ with good elements $G \subseteq X$, finds whether $|G| >  B$ using $\bigOt{B}$ independent Grover searches in $X$. This procedure uses a coupon collector instead of quantum counting~\cite{DBLP:conf/icalp/BrassardHT98}, because $B$ is considered to be a constant, and we are interested in a good success probability rather than a quadratic speedup.

\begin{lemma}
Let $X$ be a search space of exact size $2^{\alpha n}$ and $G \subseteq X$ be a ``good'' subspace of exact size unknown. There exists a quantum algorithm that given superposition access to $X$, finds whether $|G| \geq B$ or not in time $\bigOt{B \sqrt{X}}$ and with a negligible probability of error.
\end{lemma}

\begin{proof}
Although $|G|$ is not known, we use the idea (see \emph{e.g.}~\cite{brassard2002quantum}) that  we can perform quantum searches with an approximate number of iterations and still obtain solutions with a good probability.


More precisely, there exists a number $t$, depending on $|G|$ (unknown) and $|X|$ (known) such that after $t$ iterations, the state will be exactly the uniform superposition of $G$. This ideal $t$ is not an integer; it is between $1$ and $\lceil\frac{\pi}{4} \sqrt{|X|} \rceil$ (we assume that there is at least a solution, otherwise we will also detect this). Since we don't know $t$, we will instead approximate it by a $t'$ such that $\frac{t}{2} \leq t' \leq \frac{3t}{2}$. If we perform $t'$ iterations with such a good $t'$, we will fall on a state with a constant global amplitude $b$ (roughly $\frac{1}{\sqrt{2}}$) on elements of $G$. Thus, we perform many different searches with different iteration numbers, ranging from $0, 1$ to $\lceil\frac{\pi}{4} \sqrt{|X|} \rceil$, and increasing exponentially. This ensures us that regardless of the value of $t$, one of these numbers will be an approximation sufficient for us.

Since we want the algorithm to work with a global error negligible and independent of $|G|$, we set $c = \bigO{\ln \sqrt{|X|}  }$ the number of different searches and perform $c' B$ copies of each. Thus, we have a total of $c c' B$ independent states $\ket{\psi_i}$ for $1 \leq i \leq c$. One of these packets approximates the good $t$ at best, but we don't know which one. We see the state $\bigotimes_i (\ket{\psi_i}^{\otimes cB})$ as a superposition over tuples of $X^{c c' B}$:
$$ \sum_{x_1, \ldots x_{c c'B}} \alpha_{x_1, \ldots x_{c c'B}}\ket{x_1} \ldots \ket{x_{c c'B}} \enspace. $$

%

We check whether the tuple $(x_1, \ldots x_{c c' B})$ contains $\geq B$ distinct solutions and we put the result in a qubit:
$$ \sum_{x_1, \ldots x_{ c c'B}} \alpha_{x_1, \ldots x_{c c'B}}\ket{x_1} \ldots \ket{x_{c c'B}} \ket{ (x_1, \ldots x_{c c' B} \mbox{ contains $\geq B$ distinct solutions} )} \enspace. $$

%
%
%
%

If $|G| < B$, then this qubit is always $0$: we can immediately uncompute the quantum searches and we have obtained the result. If $|G| \geq B$, then some of these tuples contain $B$ distinct solutions, but not all. We must ensure that their proportion is overwhelming, so that after uncomputing, the algorithm actually adds an error vector of negligible amplitude. To do that, we will only focus on the block of $c' B$ states that corresponds to the good $t$, since for them, we have a lower bound on the probability of finding a solution. The other states, that we dismiss, can only improve our success in finding $B$ distinct solutions. So we now focus on $c'B$-tuples only.

Let us consider $|G| = B$ which is the worst case.  First, we will look at the proportion of $c'B$-tuples that contain $(1- c'')c' b B$ solutions: this is the probability to succeed at least $(1- c'')c' b B$ times after $c' B$ independent trials of probability $b$ each, and it is higher than $1 - \exp(-c'' c' B b / 2)$ by a Chernoff bound. Next, assuming that there are $(1- c'')c' b B$ independent solutions, we check the probability that they span all the $B$ distinct solutions that there are in total. This is related to the coupon collector problem. The probability to miss at least a coupon among $B$ after $c^{(3)} B \ln B$ trials is lower than $B^{-c^{(3)} + 1}$. Thus, we may take $c^{(3)} = \bigO{n}$, $c'' = \frac{1}{2}$ and $c' = \bigO{n}$ for a total probability of failure in $o(2^{-n})$. \qed
\end{proof}

\section{Computing the Fraction of Marked Vertices}
\label{appendix:probas}

This section contains a proof for Lemma~\ref{lemma:prob} that was omitted from the main body of the paper. Recall that we defined:
$$X_\vec{e}(\vec{a})= \begin{cases}
  1 & \text{ if } \vec{e}\dotprod \vec{a} = \vec{e}_0 \dotprod \vec{a} \pmod M\\
  0 & \text{ otherwise}
\end{cases}$$

We will first prove a result on the average number of vectors having the same modulus as $\vec{e}_0$, then we will use this in a Chernoff bound. Define
\[
    Y(\mathcal{B},\vec{e}_0;\vec{a})=
        \card{\set{\vec{e}\in\mathcal{B}, \vec{a}\dotprod\vec{e}=\vec{a}\dotprod\vec{e}_0\pmod{M}}}.
\]
where $M$ divides $N \simeq 2^n$. For simplicity, we write $Y(\vec{a})$ for $Y(\mathcal{B},\vec{e}_0;\vec{a})$ in the following. We are interested in $Y(\vec{a})$ as a random variable when $\vec{a}$ is drawn uniformly from $\mathbb{Z}_{N}^n$. 

%

\begin{lemma}\label{lemma:yabound}
If $|\mathcal{B}|\gg M$, then with probability $1-\negl(n)$, 
\begin{equation}
\label{eq:1}
Y(\vec{a})\leq 2\Exp{\vec{a}}{Y(\vec{a})}\sim 2\cdot \frac{|\mathcal{B}|}{M}
\end{equation}
\end{lemma}

\begin{proof}

Following \cite{NSS01}, for any $z\in\mathbb{C}$, define $\mathcal{E}(z)=\exp(2\pi iz/M)$. It satisfies the identity
\begin{equation}
\label{eq:exp}
\forall k\in \mathbb{N},
\sum_{\lambda=0}^{kM-1}\mathcal{E}(\lambda u)=
        \begin{cases}0&\text{if }u\neq 0\pmod{M}\\kM&\text{if }u=0\pmod{M}\end{cases}
\end{equation}
for any $u\in\mathbb{Z}$. We have
\[
	Y(\vec{a})=\sum_{\vec{e}\in\mathcal{B}}X_{\vec{e}}(\vec{a})
	\qquad\text{and}\qquad
    X_{\vec{e}}(\vec{a})=\frac{1}{M}\sum_{\lambda=0}^{M-1}
        \mathcal{E}(\lambda\vec{a}\dotprod(\vec{e}-\vec{e}_0)).
\]
Therefore for any $\vec{e}\neq\vec{e}_0$,
\begin{align*}
	\Exp{\vec{a}}{X_{\vec{e}}(\vec{a})}
	&=\frac{1}{N^{n}}\sum_{\vec{a}\in \mathbb{Z}_{N}^n} \frac{1}{M}
		\sum_{\lambda=0}^{M-1}
		\mathcal{E}(\lambda\vec{a}\dotprod(\vec{e}-\vec{e}_0))\\
	&=\frac{1}{N^{n}}\sum_{\vec{a}\in \mathbb{Z}_{N}^n} \frac{1}{M}
	+\frac{1}{N^{n}}\sum_{\vec{a}\in \mathbb{Z}_{N}^n}\frac{1}{M}
		\sum_{\lambda=1}^{M-1}
		\mathcal{E}(\lambda\vec{a}\dotprod(\vec{e}-\vec{e}_0))\\
	&=\frac{1}{M}+\frac{1}{MN^{n}}\sum_{\lambda=1}^{M-1}
		\sum_{\vec{a}\in \mathbb{Z}_{N}^n}
	 	\mathcal{E}(\lambda\vec{a}\dotprod(\vec{e}-\vec{e}_0))\\
	&=\frac{1}{M}+\frac{1}{MN^{n}}\sum_{\lambda=1}^{M-1}
			\prod_{i=1}^n\sum_{a\in \mathbb{Z}_{N}}
			\mathcal{E}(\lambda a(\vec{e}^i-\vec{e}_0^i)))\\
&=\frac{1}{M}&&\text{(see below)}
\end{align*}
where in the last step, we used that if $\vec{e}\neq\vec{e}_0$,
there exists $i\in [1,n]$ such that $\vec{e}^i\neq\vec{e}_0^i$,
where $\vec{e}^i$ is the $i$th component of $\vec{e}$ and
hence of the $i^{th}$ sum in the product is zero by \eqref{eq:exp}
since $\lambda\neq 0\pmod M$. It follows by linearity that
\[
	\Exp{\vec{a}}{Y(\vec{a})}
		=\Exp{\vec{a}}{X_{\vec{e}}(\vec{a})}
			+\sum_{\vec{e}\in\mathcal{B}\setminus\{\vec{e}_0\}}
				\Exp{\vec{a}}{X_{\vec{e}}(\vec{a})}
		=1+\frac{|\mathcal{B}|-1}{M}
\]
since $\vec{e}_0\in\mathcal{B}$.
Similarly for any
$\vec{e},\vec{f}\in\mathcal{B}\setminus\{\vec{e}_0\}$,
\begin{align*}
	\Exp{\vec{a}}{X_{\vec{e}}(\vec{a})X_{\vec{f}}(\vec{a})}
	&=\frac{1}{N^{n}}\sum_{\vec{a}\in \mathbb{Z}_{N}^n}
		\left(\frac{1}{M}\sum_{\lambda=0}^{M-1}
			\mathcal{E}(\lambda\vec{a}\dotprod(\vec{e}-\vec{e}_0)
		)\right)
		\left(\frac{1}{M}\sum_{\mu=0}^{M-1}
			\mathcal{E}(\mu\vec{a}\dotprod(\vec{f}-\vec{e}_0)
		)\right)\\
	&=\frac{1}{M^2N^{n}}\sum_{\vec{a}\in \mathbb{Z}_{N}^n}
		\sum_{\lambda=0}^{M-1}\sum_{\mu=0}^{M-1}
			\mathcal{E}(\lambda\vec{a}\dotprod(\vec{e}-\vec{e}_0))
			\mathcal{E}(\mu\vec{a}\dotprod(\vec{f}-\vec{e}_0))\\
	&=\frac{1}{M^2N^{n}}\sum_{\vec{a}\in \mathbb{Z}_{N}^n}
		\sum_{\lambda=0}^{M-1}\sum_{\mu=0}^{M-1}
		\mathcal{E}(\vec{a}\dotprod(\lambda\vec{e}+\mu\vec{f}-(\lambda+\mu)\vec{e}_0))\\
	&=\frac{1}{M^2N^{n}}
		\sum_{\lambda=0}^{M-1}\sum_{\mu=0}^{M-1}
		\sum_{\vec{a}\in \mathbb{Z}_{N}^n}
		\mathcal{E}(\vec{a}\dotprod(\lambda\vec{e}+\mu\vec{f}-(\lambda+\mu)\vec{e}_0))\\
	&=\frac{1}{M^2}
		\sum_{\lambda=0}^{M-1}\sum_{\mu=0}^{M-1}
		\mathds{1}\big[
			\lambda\vec{e}+\mu\vec{f}=(\lambda+\mu)\vec{e}_0\mod{M}
			\big]
\end{align*}
by \eqref{eq:exp}.
If $\lambda=0$ then the equation $\lambda\vec{e}+\mu\vec{f}=(\lambda+\mu)\vec{e}_0\mod{M}$
becomes $\mu\vec{f}=\mu\vec{e}_0$ but since $\vec{f}\neq\vec{e}_0$, the only solution is $\mu=0$.
A symmetric reasoning shows that if $\mu=0$ then $\lambda=0$
is the only solution. Hence, given $\vec{e}\neq\vec{e}_0$, we have
\[
	\sum_{\vec{f}\in\mathcal{B}\setminus\{\vec{e}_0\}}
		\Exp{\vec{a}}{X_{\vec{e}}(\vec{a})X_{\vec{f}}(\vec{a})}
		=\frac{|\mathcal{B}|-1+|F_{\vec{e}}|}{M^2}
\]
where
\[
	F_{\vec{e}}=\left\{(\lambda,\mu,\vec{f})\in
		\{1,\dots,M-1\}^2\times\mathcal{B}\setminus\{\vec{e}_0\}:
		\lambda\vec{e}+\mu\vec{f}=(\lambda+\mu)\vec{e}_0\mod{M
		}\right\}.
\]
We now claim that this set is not too large.
Assume that $(\lambda,\mu,\vec{f})\in F_{\vec{e}}$, recall
that $\lambda,\mu\neq0$ and
since $\vec{e}\neq\vec{e}_0$ then
there exists $i$ such that $\vec{e}^i\neq\vec{e}_0^i$ so in 
particular
$\lambda(\vec{e}^i-\vec{e}_0^i)=\mu(\vec{f}^i-\vec{e}_0^i)$.
But recall that
$\vec{e},\vec{f},\vec{e}_0\in\mathcal{B}\subseteq\{-1,0,1,2\}^n$
hence $\vec{e}^i-\vec{e}_0^i\in\{-3,-2,-1,1,2,3\}$.
It follows that if we fix $\mu$ then there are at most $3$ possible
values\footnote{If we have, say, $3\lambda=x\pmod{M}$ then
$\lambda=x/3\pmod{M/3}$, which is only possible if $x$ and $M$
are divisible by $3$, and then $\lambda\in\{x/3,(x+M)/3,(x+2M)/3\}$.}
for $\lambda$. We note in passing that the constant $3$ is not
magical: if we had $\mathcal{B}\subseteq\{-a,\ldots,a\}^n$ then it be
bounded by $2a$. Now assume that
$(\lambda,\mu,\vec{f}),(\lambda,\mu,\vec{g})\in F_{\vec{e}}$
with $\vec{f}\neq\vec{g}$, then we must have
\begin{align*}
	\mu(\vec{f}-\vec{g})=0\mod M
	&\qquad\Rightarrow\qquad \vec{f}-\vec{g}=0\mod M/\mu\\
	&\qquad\Rightarrow\qquad\exists k\neq 0.
		\forall i, \vec{g}^i=\vec{f}^i+kM/\mu
\end{align*}
which is only possible if $\mu$ divides $M$. In particular, we must have
$M/\mu\geqslant2$, in other words all coordinates of $\vec{f}$
and $\vec{g}$ are at least at distance $2$. This is clearly impossible because $\mathcal{B}\subseteq D^n[\alpha,\beta,\gamma]$:
the distribution of ``-1'', ``0'', ``1'', ``2'' in one of $\vec{f}$ or
$\vec{g}$ would be wrong. In summary, we have that:
\begin{itemize}
	\item for every $\mu$, $\vec{f}$, there are at most
		3 possible values of $\lambda$ such that
		$(\lambda,\mu,\vec{f})\in F_{\vec{e}}$,
	\item for every $\lambda$ and $\mu$, there is at most one value
		of $\vec{f}$ such that $(\lambda,\mu,\vec{f})\in F_{\vec{e}}$.
\end{itemize}
It follows that $F_{\vec{e}}$ has size at most $3M$.
Then by linearity,
\begin{align*}
	\Exp{\vec{a}}{Y(\vec{a})^2}
		&=\sum_{\vec{e},\vec{f}\in\mathcal{B}}\Exp{\vec{a}}{
			X_{\vec{e}}(\vec{a})X_{\vec{f}}(\vec{a})}\\
		&=\sum_{\vec{f}\in\mathcal{B}}
			\Exp{\vec{a}}{X_{\vec{f}}(\vec{a})}
			+\sum_{\vec{e}\in\mathcal{B}\setminus\{\vec{e}_0\}}
			\Exp{\vec{a}}{X_{\vec{e}}(\vec{a})}
			+\sum_{\vec{e},\vec{f}\in\setminus\{\vec{e}_0\}}\Exp{\vec{a}}{
			X_{\vec{e}}(\vec{a})X_{\vec{f}}(\vec{a})}\\
		&\leqslant2\Exp{\vec{a}}{Y(\vec{a})}-1
			+(|\mathcal{B}|-1)\frac{|\mathcal{B}|-1+3M}{M^2}\\
		&\leqslant1+2\frac{|\mathcal{B}|-1}{M}
			+(|\mathcal{B}|-1)\frac{|\mathcal{B}|-1+3M}{M^2}\\
		&\leqslant
			\frac{M^2+(|\mathcal{B}|-1)(|\mathcal{B}|-1+5M)}{M^2}.
\end{align*}
Finally, we have
\begin{align*}
        \mathbb{V}_{\vec{a}}(Y(\vec{a}))
            &=\Exp{\vec{a}}{Y(\vec{a})^2}-\Exp{\vec{a}}{Y(\vec{a})}^2\\
            &\leqslant\frac{M^2+(|\mathcal{B}|-1)(|\mathcal{B}|-1+5M)
            	-(|\mathcal{B}|+M-1)^2}{M^2}\\
		&=\frac{3(|\mathcal{B}|-1)M}{M^2}\\
		&=\frac{3(|\mathcal{B}|-1)}{M}.
\end{align*}
Thus $\Exp{\vec{a}}{Y(\vec{a})}\approx \mathbb{V}_{\vec{a}}(Y(\vec{a}))$ when we look at their order of magnitude.

According to Tchebychev's inequality, $$\prob[\vec{a}]{|Y(\vec{a})-\Exp{\vec{a}}{Y(\vec{a})}|>\Exp{\vec{a}}{Y(\vec{a})}}\leq \frac{\mathbb{V}_{\vec{a}}(Y(\vec{a}))}{\Exp{\vec{a}}{Y(\vec{a})}^2} =\negl(n)$$
which completes the proof. \qed
\end{proof}

\setcounter{lemma}{6}
\begin{lemma}
If $|\mathcal{B}|\gg M \simeq |L|$, then for a $1-\negl(n)$ proportion of $\vec{a} \in \mathbb{Z}_{N}^n$, and with an appropriate $B = \bigO{n}$:
\begin{equation}
\prob[\vec{e}_1, \cdots, \vec{e}_{|L|} \sim Unif(\mathcal{B})]{\sum_{i=1}^{|L|} X_{\vec{e}_i} (\vec{a})< B -1} > 1 - \frac{1}{\poly(n)}
\end{equation}
\end{lemma}

\begin{proof}

\begin{multline*}
\prob[\vec{e}_1, \cdots, \vec{e}_{|L|} \sim Unif(\mathcal{B})]{\sum_{i=1}^{|L|} X_{\vec{e}_i} (\vec{a})< B -1}  \\
=  \sum_{y=1}^{|\mathcal{B}|} \prob[\vec{e}_1, \cdots, \vec{e}_{|L|} \sim Unif(\mathcal{B})]{\sum_{i=1}^{|L|} X_{\vec{e}_i} (\vec{a})<B -1 | Y(\vec{a})=y} \prob{Y(\vec{a})=y}\\
\end{multline*}

Under the condition of $Y(\vec{a})=y$, for all $i \in [1,|L|]$, $X_{\vec{e}_i}(\vec{a})$ can be seen as a random variable following $\Ber(\frac{y}{|\mathcal{B}|})$, since $\vec{e}_i$s' are randomly chosen from $\mathcal{B}$. Here $\Ber(p)$ is a Bernoulli distribution of parameter $p$. 

Using equation (\ref{eq:1}), for a $1-\negl(n)$ portion of $\vec{a}\in \mathbb{Z}_{N}^n$, we have
\begin{align*}
&\prob[\vec{e}_1, \cdots, \vec{e}_{|L|} \sim Unif(\mathcal{B})]{\sum_{i=1}^{|L|} X_{\vec{e}_i} (\vec{a})<B -1}\\
& >  \sum_{y=1}^{2\cdot\frac{|\mathcal{B}|}{M}} \prob[\vec{e}_1, \cdots, \vec{e}_{|L|} \sim Unif(\mathcal{B})]{\sum_{i=1}^{|L|} X_{\vec{e}_i} (\vec{a})<B -1 | Y(\vec{a})=y} \prob{Y(\vec{a})=y}    \\
& =  (1-\negl(n))\sum_{y=1}^{2\cdot\frac{|\mathcal{B}|}{M}} \prob[\vec{e}_1, \cdots, \vec{e}_{|L|} \sim Unif(\mathcal{B})]{\sum_{i=1}^{|L|} X_{\vec{e}_i} (\vec{a})<B -1 | Y(\vec{a})=y} \\
& >(1-\negl(n))\prob[\vec{e}_1, \cdots, \vec{e}_{|L|} \sim Unif(\mathcal{B})]{\sum_{i=1}^{|L|} X_{\vec{e}_i} (\vec{a})<B-1| Y(\vec{a})= 2\cdot \frac{|\mathcal{B}|}{M}}\\
& = (1-\negl(n))\prob[X_{\vec{e}_i}\sim \Ber(2\cdot  \frac{|\mathcal{B}|}{M}\cdot \frac{1}{|\mathcal{B}|})=\Ber(\frac{2}{M})]{\sum_{i=1}^{|L|} X_{\vec{e}_i} (\vec{a})<B-1}
\end{align*}

Chernoff's inequality gives that for any $\delta \geq 1$:
$$ \prob{ \sum_{i=1}^{|L|} X_{\vec{e}_i} (\vec{a}) \geq (1+\delta) \frac{2|L|}{M} } \leq e^{ - \frac{\delta}{3} \frac{2|L|}{M} } \enspace.$$

Hence, when $M = |L|$, by taking $B$ linear in $n$, we obtain that the probability of being unmarked due to this $\vec{e_0}$ is less than $\frac{1}{\poly(n)}$. \qed
\end{proof}


%

\end{document}